\documentclass[journal]{IEEEtran}
\IEEEoverridecommandlockouts    
\usepackage{syntonly, stfloats}
\usepackage{enumerate}
\usepackage{graphicx,epsfig}
\usepackage{amsmath,amssymb,amsthm}   
\usepackage{hyperref,url}
\usepackage{harpoon}
\usepackage{mathtools}
\usepackage{latexsym, array}
\newtheorem{theorem}{Theorem}[section]

\begin{document}
\title{Rollover Preventive Force Synthesis at Active Suspensions in a Vehicle Performing a Severe Maneuver with Wheels Lifted off}
\author{Nigam Chandra Parida, \thanks{SERC, Indian Institute of Science, Bangalore 560012, India;  email: {\it ncparida@gmail.com}} 
 Soumyendu Raha \thanks{SERC, Indian Institute of Science, Bangalore 560012, India;  email: {\it raha@serc.iisc.in}
} and Anand Ramani \thanks{Visitor to SERC, Indian Institute of Science, Bangalore; email: {\it ramani.anand@gmail.com}} }

 \maketitle
\begin{abstract}
Among the intelligent safety technologies for road vehicles, active suspensions controlled by 
embedded computing elements for preventing rollover have received a lot of attention. 
The existing models for synthesizing and allocating forces in such suspensions are conservatively based on
the constraint that no wheels lift off the ground. However, in practice, smart/active suspensions are more 
necessary in the situation where the wheels have just lifted off the ground. The difficulty in computing control in the last
situation is that the problem requires satisfying disjunctive constraints on the dynamics. 
To the authors\textquoteright\,knowledge, no efficient solution method is available for the  simulation of dynamics
with disjunctive constraints and thus hardware realizable and accurate force allocation in an active suspension tends to be 
a difficulty. In this work we give an algorithm for and simulate numerical solutions of the force allocation problem as an
optimal control problem constrained by dynamics with disjunctive constraints. 
In particular we study the allocation and synthesis of time-dependent active suspension forces 
in terms of sensor output data  in order to stabilize the roll motion of the road vehicle. 
An equivalent constraint in the form of a convex combination (hull) is proposed to 
satisfy the disjunctive constraints. The validated numerical simulations show that it is possible to allocate
and synthesize control forces at the active suspensions from sensor output data such that the forces stabilize
the roll moment of the vehicle with its wheels just lifted off the ground during arbitrary fish-hook maneuvers.
\end{abstract}
\section{Introduction}
Traditionally forces allocated in intelligent suspension systems of a road vehicle to prevent its rollover during a severe maneuver 
have been based on the conservative constraint that 
none of its wheels  has lifted off the road surface (see, for example, models in \cite{rajamaniBook} and \cite{ramani}; see \cite{ramani} for a detailed 
recent literature review which we omit here to avoid repetition). However, this condition may not be practical in an 
actual situation where the vehicle does a severe fish-hook maneuver to avoid a sudden obstacle around a corner or on a  relatively tight 
curve, when there is no sufficient space and time left for slowing down. Such maneuvers are well known
to produce large yaw rates that induce rollover by lifting the wheels on one side of the vehicle off the road surface. 
The more severe problem of stabilization when
wheels have lifted off yields disjunctive constraints (i.e., either sufficient anti-roll moment to the left or to the right, depending
on which wheels are lifted off, must be available) on the vehicle dynamics. It then 
becomes necessary to solve a disjunctively constrained dynamical optimization problem to obtain the stabilizing forces 
that must be synthesized in the suspensions. 
In this work we formulate and solve a disjunctively constrained dynamical optimization problem and in the process, find 
forces in the suspensions that would assist in preventing the  rollover of a vehicle in the more severe situation of wheels
just lifting off the ground.
\par
Disjunctively constrained dynamics has not been studied often, although 
general disjunctive programming with nonlinear algebraic constraints have
received sufficient attention in recent years. A review of methods of 
handling disjunctive constraints in non-linear optimization problems where constraints do not include dynamics, 
can be found in \cite{grossmanReview} and in Part II of \cite{imaminlp}. 
\par
In computation schemes for collision avoidance (e.g., intelligent transportation system involving many vehicles) \cite{ak1, ak2, abi}, 
disjunctive constraints on the dynamics are converted into more conservative negated conjunctive constraints on the critical section 
in which collision is to be avoided and the schemes use large numbers and Kronecker deltas.  
A drawback of this approach (called ``big-M constraints'' in \cite{grossmanReview}) method 
is that the computational relaxation affected is often weak \cite{grossmanReview} resulting in failure of the 
disjunctive program. In the rollover prevention problem, it is more safety critical to 
satisfy the disjunctive constraints tightly.
\par
Thus  a general but efficient approach must be developed for computing the force allocation in suspensions correctly 
in the context of the rollover prevention problem. A convex hull (outlined in  \cite{grossmanReview} and
 in Part II of \cite{imaminlp}) of the the functions which enter the disjunctive constraints 
is used in this work to compute the effective force allocations in the active suspensions of the vehicle.
\par
The paper is organized as follows. At first the mathematical model of a road vehicle undergoing a severe maneuver 
along with different constraints is described, followed  by a method to handle the disjunctive constraints
that allow the wheels of the vehicle being lifted off the road surface. Section \ref{sec:GAlpha} describes 
the discretization of the disjunctively constrained dynamics. Section \ref{SecTrans} outlines the direct transcription method of
solving the resulting optimal control problem with disjunctive dynamics as constraint. 
Comparison of the disjunctive dynamics approach with the more conservative conjunctive constraints 
is illustrated with numerical results in section \ref{sec:Numerical}.
The rollover preventive forces with wheels lifted off are computed as sensor adapted controls in section 
\ref{sec:Numerical} to show that it is possible to synthesize the desired controls in terms of yaw rate using the present disjunctive dynamics model.
\section{Main Results}
The main results of this work are: 
(a) disjunctive dynamics constrained model of  rollover of a road vehicle undergoing a severe maneuver in which wheels on one side are allowed 
to be lifted off  the ground and (b) computation, using this model, to find 
the control forces that should be synthesized in an active suspension system in order to stabilize the  roll of the vehicle even when
wheels on one side have just lifted off the ground. The control forces are shown to be obtained in terms of the yaw rate sensor output.
\par
The maneuver indicated by figure  \ref{Fig:SteeringFun} is a fish-hook maneuver in which wheels are allowed to
be lifted off the ground and the control forces in the active suspension are generated based on yaw rate sensor output
to stabilize the vehicle. It is found  that the synthesized control forces make the vertical reactions on the right wheels zero during 
the course of the maneuver in which wheels lift off, and the resulting roll moment is negative, 
indicating that the rollover tendency of the vehicle with wheels lifted-off is neutralized.
The numerical solutions show that  satisfying the disjunctive constraints is key to computing the anti-roll controls.
\section{Mathematical and Computational Preliminaries}
In this section we layout the mathematical and computational tools and approach that are the bases of our
vehicle rollover model with disjunctively constrained dynamics and subsequent force synthesis.
\subsection{Dynamic Optimization Problem}
We approach the present control synthesis problem in the following dynamic optimization form:
\begin{subequations}
\begin{eqnarray}
&& \min ~~\int^{t_f}_{t_0} L \left(x,u, t \right) dt  ~~ \mbox{(objective function)} \label{optima1}\\
&& \mbox{subject to} \nonumber \\
&& \dot{x}=f(x,u,t)~~~\mbox{(dynamics)} \label{dyn} \\
&& 0 \le \phi(x,t)~~~\mbox{(inequality path constraints)} \label{cons} \\
&& 0 = \tilde{\phi}(x,t)~~\mbox{(equality path constraints)} \label{econs} \\
&&  \bigvee_{i=1}^{m_d} (0 \le \varphi_i) ~ \mbox{(inclusive disjunction on constraints)} \label{disjuncs}\\
&& b_{upper} \ge {\begin{pmatrix} x \\ u \end{pmatrix}} \ge b_{lower} ~\mbox{(bounds)}~\mbox{and} \label{bounds} \\
 && x(0), u(0) ~\mbox{are consistent initial conditions} \label{init}
\end{eqnarray}
\label{probdef}
\end{subequations}
where $x \in {\mathbb R}^m$ denotes the differential state variables, 
$f: {\mathbb R}^m \times {\mathbb R}^{q_u}\times {\mathbb R}\rightarrow {\mathbb R}^m $,
$u \in {\mathbb R}^{q_u} $ is the vector of controls and algebraic 
state variables; 
$\phi: {\mathbb R}^{m} \times {\mathbb R}  \rightarrow {\mathbb R}^{(q_e)} $ is the vector of inequality path constraints; 
$\varphi_i: {\mathbb R}^{m} \times {\mathbb R}  \rightarrow {\mathbb R}^{q_d},~i=1 \cdots m_d$ are the constraints over
which disjunction is taken, $\tilde{\phi}:{\mathbb R}^{m} \times {\mathbb R}  \rightarrow {\mathbb R}^{q_r} $ 
is the vector of equality path constraints and $t \in [t_0,t_f] \subset {\mathbb R}$ with $t_f > t_0$.  In the above $q_e, q_d, q_r$ are such that
no more than $q_u$ constraints are active at any given time in the simulation interval.
\subsection{Disjunctive Constraints as Convex Hull}\label{SecEO} 
Disjunctive constraints can be incorporated as convex constraints. This is equivalent to representing the disjunction  as
a convex hull of the constraints entering the disjunction. An inclusive disjunction or inclusive logical or over 
the functions $f_i: {\mathbb R}^n \to {\mathbb R},~ i \in \{1, \cdots, m\}$ can be represented
as follows:
\begin{eqnarray}
&& \bigvee_{i=1}^{m_d} (f_i(x) \le 0)  \Leftrightarrow \sum_{i=1}^{m_d} \lambda_i f_i(x) \le 0, ~\mbox{such that} \nonumber\\
&& \sum_{i=1}^{m_d} \lambda_i = 1,~\lambda_i \ge 0 ~\forall~ i \in \{1,\cdots,m_d\} 
\label{cvior}
\end{eqnarray}
 From the construction of (\ref{cvior}) it is clear that 
 $\{\lambda_i |   \sum_{i=1}^{m_d} \lambda_i = 1,~\lambda_i \ge 0 ~\forall~ i \in \{1,\cdots,m_d\} \}$
 can be found if at least one of the $f$'s is non-positive. 
 If $f_i(x) \ge 0, \forall~ i \in \{1,\cdots,m_d\}$, i.e., 
 all of the functions, over which the inclusive disjunction is specified, 
 are positive, then no $\lambda$'s can be found and the inclusive disjunction is correctly indicated as infeasible. 
 \par
 It follows that exclusive disjunction or exclusive-or between two constraints, i.e., 
 $(f_1 \le 0) \oplus (f_2 \le 0)$ can be represented by the following constraints:
 $\lambda_1 f_1 + \lambda_2 f_2 \le 0,~\pi_1 f_1 + \pi_2 f_2 \ge 0,~\lambda_1 + \lambda_2 = 1, ~ \pi_1 + \pi_2 = 1,
 \pi_i \ge 0,~\lambda_i \ge 0,~i=1,2$. 
 \par
In the context of the constrained dynamics, the inclusive disjunction on constraints at some $t \in I\subseteq [t_0,t_f]$ implies that 
at least one of the constraints (over which disjunction is taken) is satisfied. 
Consider $f_1(t,X(t))$ and $f_2(t,X(t))$ to be two constraint functions. 
In each sub-interval $I$ of $[t_0,t_f]$ if $f_1(t,X(t))> 0$ then $f_2(t,X(t))\leq 0$ must hold and vice versa.
The inclusive disjunctive constraint is also satisfied, when, as appropriate, both $f_1$ and $f_2$ are non-positive.
We implement the above by satisfying $g(t,X(t)):= \lambda(t) f_1(t,\,X(t))+\left(1-\lambda(t)\right)f_2(t,\,X(t))\leq 0$ 
such that $\lambda(t) \in [0,1]$.
\begin{theorem}
Consider the inclusive disjunctive constraint at a $t \in [t_0,t_f]$: $(f_1(t,X(t))\leq 0) \bigvee (f_2(t,X(t))\leq 0)$. 
If and only if $X(t)$ satisfies $g(t,X(t)):= \lambda(t) f_1(t,\,X(t))+\left(1-\lambda(t)\right)f_2(t,\,X(t))\leq 0$ for some 
$\lambda(t) \in [0,1]$, the inclusive disjunction is satisfied.
\end{theorem}
\begin{proof}
Suppose at least one of $(f_1(t,X(t))\leq 0)$ and $(f_2(t,X(t))\leq 0)$ is true at $t \in [t_0,t_f]$. 
By construction, we can choose a $\lambda(t)\in[0,\,1]$ for $t\in[t_0,\,t_f]$
so that $g(t,\,X(t)) \le 0$. A $\lambda(t)$ cannot be found only when both $f_1$ and $f_2$ are positive.
The converse is as follows.  Let $g(t,X(t))\leq 0$ hold for some $\lambda(t) \in [0,1]$.  Since the sum of two positive reals 
can not be negative, $g(t,X(t))\leq 0$ implies that 
\begin{eqnarray}
&& \mbox{either}~\lambda(t) f_1(t,\hat{X}(t))\leq 0 ~ \mbox{or} 
\left(1-\lambda(t)\right)f_2(t,\hat{X}(t))\leq 0, ~ \mbox{or}, \nonumber \\ 
&& \mbox{both $f_1(t,\hat{X}(t))\leq 0$ and $f_2(t,\hat{X}(t))\leq 0$}.\label{EqEo2} 
\end{eqnarray}
Hence the claim follows.
\end{proof}
We note that the variable $\lambda(t)$ used to handle the inclusive disjunctive constraints 
does not have a unique solution. However, from the computational
point of view $\lambda(t)$ is treated as an algebraic variable (in the context of a differential-algebraic
equation model of the vehicle rollover dynamics). Equations (\ref{EO1}) and (\ref{EO2}) in section \ref{ConsSec} are based on this approach.
\subsection{Numerical Algorithm for Solving the Dynamic Optimization Problem}
In this section we give a step-wise algorithm for solving  the dynamic optimization formulation (\ref{probdef}) of the control problem.
\begin{enumerate}[\hspace{0.5cm}\bfseries Step 1]
\item The  ODE (converted to canonical first order) describing the dynamics over the entire time interval of the maneuver is discretized.
\item As described in section \ref{SecEO}, the convex hull equivalent of the disjunctive constraints is appended to the discretized dynamics.
\item The finite dimensional dynamic optimization problem thus formed is solved by the direct transcription method described in 
section \ref{SecTrans} using a nonlinear programming solver. 
\end{enumerate}
\subsubsection{The $\alpha$-Method Discretization of the Dynamics}\label{sec:GAlpha}
For definiteness, we use the $\alpha$-method to disceretize the dynamics in our numerical method 
(cf. \cite{GAlpha, Multi} for first order DAEs,  \cite{DaeAlpha} for second order DAEs). The method has the property of 
producing regularized (reduced condition number) constraint Jacobians and its DAE discretization is unconditionally stable (cf. \cite{GAlpha} for 
mathematical theory and computational properties when used in a direct transcription). These
two properties are useful because of the stiffness the disjunctive constraints produce by the switching action inherent in the 
disjunction. The method discretizes
the first order initial value problem $\dot x=f(x,t),\,g(x,u,t)=0,\,x(t_0)=x_0$ as  
\begin{subequations}
\begin{eqnarray}
x_{n+1}&=&x_n+\left(1-\frac{\beta}{\gamma}\right)h_nf(x_n,\,t_n)
        +\frac{\beta}{\gamma}h_nf(x_{n+1},t_{n+1})\nonumber\\
       & &+\left(\frac{1}{2}-\frac{\beta}{\gamma}\right)h_n^2a_n\\
a_{n+1}&=&\frac{1}{h_n\gamma}\left(f(x_{n+1},t_{n+1})-f(x_n,t_n)\right) 
+\left(1-\frac{1}{\gamma}\right)\\
0 & = & g(x_{n+1}, u_{n+1},t_{n+1})
\end{eqnarray}
\end{subequations}
where $h_n=t_{n+1}-t_n$ is the time step size, and $a_0$ is either given or calculated by 
$a_0=\frac{\mathrm{d}f}{\mathrm{d}t}$ at $t=0$.
 The parameters $\gamma$ and $\beta$ are computed as
 \begin{subequations}
\begin{eqnarray}
\gamma&=&\frac{2}{\rho+1}-\frac{1}{2} \\
\beta&=&\frac{1}{(\rho+1)^2}
\end{eqnarray}
\end{subequations}
where $\rho \in \left[0,\,1\right)$ is a user-selected variable. 
\subsubsection{Direct Transcription Method}\label{SecTrans}
Consider the following optimal control problem.
\begin{subequations}
\begin{eqnarray}
\min J&=&\int_{t_0}^{t_f}L(x,\,u,\,t)\label{GOCP1}\\
\textnormal{subject to}& &\nonumber\\
&& \dot x = f(x,\,u,\,t),\,x(t_0)=x_0\label{GOCP2}\\
&& g(x,\,u,\,t) = 0\label{GOCP3}\\
&& h(x,\,u,\,t) \leq 0\label{GOCP4}
\end{eqnarray}
\end{subequations}
where the initial value $u(t_0)=u_0$ may or may not be given. 
This is an infinite-dimensional continuous problem over $[t_0,t_f]$. We approximate the problem by a
finite-dimensional version by discretizing the dynamics over $[t_0, t_f] $ partitioned
as $t_0<t_1<\cdots<t_N=t_f$. The objective function 
$J$ is approximated by the trapezoidal rule whereas the dynamics is discretized 
by the $\alpha$-method and the constraints are required to be satisfied at each grid point. 
We assume that $u(t_0)=:u_0$ is either known or can be computed such that $x_0,\,u_0$ are consistent with the 
equality and inequality constraints (\ref{GOCP3}) and (\ref{GOCP4}).  The finite dimensional discretized problem is then written as 
\begin{subequations}
\begin{eqnarray}
&& \min \sum_{i=0}^{i=N-1}\frac{t_{i+1}-t_i}{2} \nonumber \\ &&
 \times \left(L(x_i,\,u_i,\,t_i)+L(x_{i+1},\,u_{i+1},\,t_{i+1})\right)\label{Trans1} \\
&& \mbox{subject to, for $n=0,\,1,\cdots, N-1$ and $h_n=t_{n+1}-t_n$,} \nonumber \\
&& x_{n+1}=x_n+\left(1-\frac{\beta}{\gamma}\right)h_nf(x_n,\,t_n)
        +\frac{\beta}{\gamma}h_nf(x_{n+1},t_{n+1})\nonumber\\
       &&+\left(\frac{1}{2}-\frac{\beta}{\gamma}\right)h_n^2a_n\label{Trans2}\\
&& a_{n+1}=\frac{1}{h_n\gamma}\left(f(x_{n+1},t_{n+1})-f(x_n,t_n)\right)\label{Trans3}\\
&& 0=g(x_{n+1},\,u_{n+1},\,t_{n+1})\label{Trans4}\\
&& 0\geq h(x_{n+1},\,u_{n+1},\,t_{n+1})\label{Trans5}\\
&& x(t_0)=x_0;~0=g(x_0,\,u_0,\,t_0); ~0\geq h(x_0,\,u_0,\,t_0)\label{Trans6}
\end{eqnarray}
\end{subequations} 
in which $a_0$ is also computed when not known.
The problem (\ref{Trans1}-\ref{Trans6}) is solved by a nonlinear programming (NLP) solver, 
such as, a sequential quadratic programming (SQP) method. 
For a more detailed description of the method, the reader is referred to \cite{GAlpha}.
\section{Mathematical Model of the Roll Stabilization}\label{MathModel}
The vehicle dynamics model (cf. \cite{ramani} for description) is described as a constraint in an optimal control problem. The model, with reference to figures \ref{vehmod1} and \ref{vehmod2} and with reference to the parameters described in the 
Appendix, is given as \begin{subequations} \begin{eqnarray}
&& \mbox{Minimize } J= \nonumber \\
&& \int_0^{t_f}\left\{\left[X\left(t\right)- \overline X\left(t\right)\right]^2 +
    \left[Y\left(t\right)- \overline Y\left(t\right)\right]^2\right\}\,\mathrm{d}t \qquad \label{Mobj}\\
&& \textnormal{subject to the following equations of motion:}\nonumber\\
&& M\ddot{X}=\sum_{i=1}^4\mu F_{Yi}\sin \left(\theta_Z+\delta_i\right)\label{EqM1}\\
&& M\ddot{Y}=\sum_{i=1}^4-\mu F_{Yi}\cos \left(\theta_Z+\delta_i\right)\label{EqM2}\\
&& I_{ZZ}\ddot\theta_Z=\sum_{i=1}^4\big(-\mu F_{Yi}\cos\delta_i r_{Xi} \nonumber\\ &&-
         \mu F_{Yi}\sin\delta_i r_{Yi}\big)\label{EqM3}\\
&& M\ddot{Z}=\sum_{i=1}^4F_{Zi}-Mg\label{EqM4}\\
&& I_{XX}\ddot\theta_X=\left(F_{Z1}-F_{Z2}+F_{Z3}-F_{Z4}\right)\frac{T}{2}\nonumber\\
                     & &+\left(F_{Z1}+F_{Z2}+F_{Z3}+F_{Z4}\right)Z\tan\theta_X\nonumber\\
                     & &+MZ\left(\ddot{Y}\cos\theta_Z-\ddot{X}\sin{\theta_Z}\right), \label{EqM5}\end{eqnarray}
\label{RollStabModel} \end{subequations}
to which we shall append the roll stabilization specific constraints in section \ref{ConsSec}. Other constraints are 
bounds from suspension travel limits, force limit, etc..
In model (\ref{RollStabModel}), the reference quantities $\overline X(t)$ and $\overline Y(t)$ are obtained 
by solving the state equilibrium equations (\ref{EqM1}-\ref{EqM3}) after setting 
$F_l=0,\,F_r=0,\,Z=0,\,\dot{Z}=0,\,\theta_X=0,\,\dot{\theta}_X=0$ and the objective function (\ref{Mobj}) 
makes the vehicle follow the reference path $\left(\overline X(t),\,\overline Y(t)\right)$ as closely as possible.
\begin{figure} 
 \centering
 \includegraphics[scale=0.33]{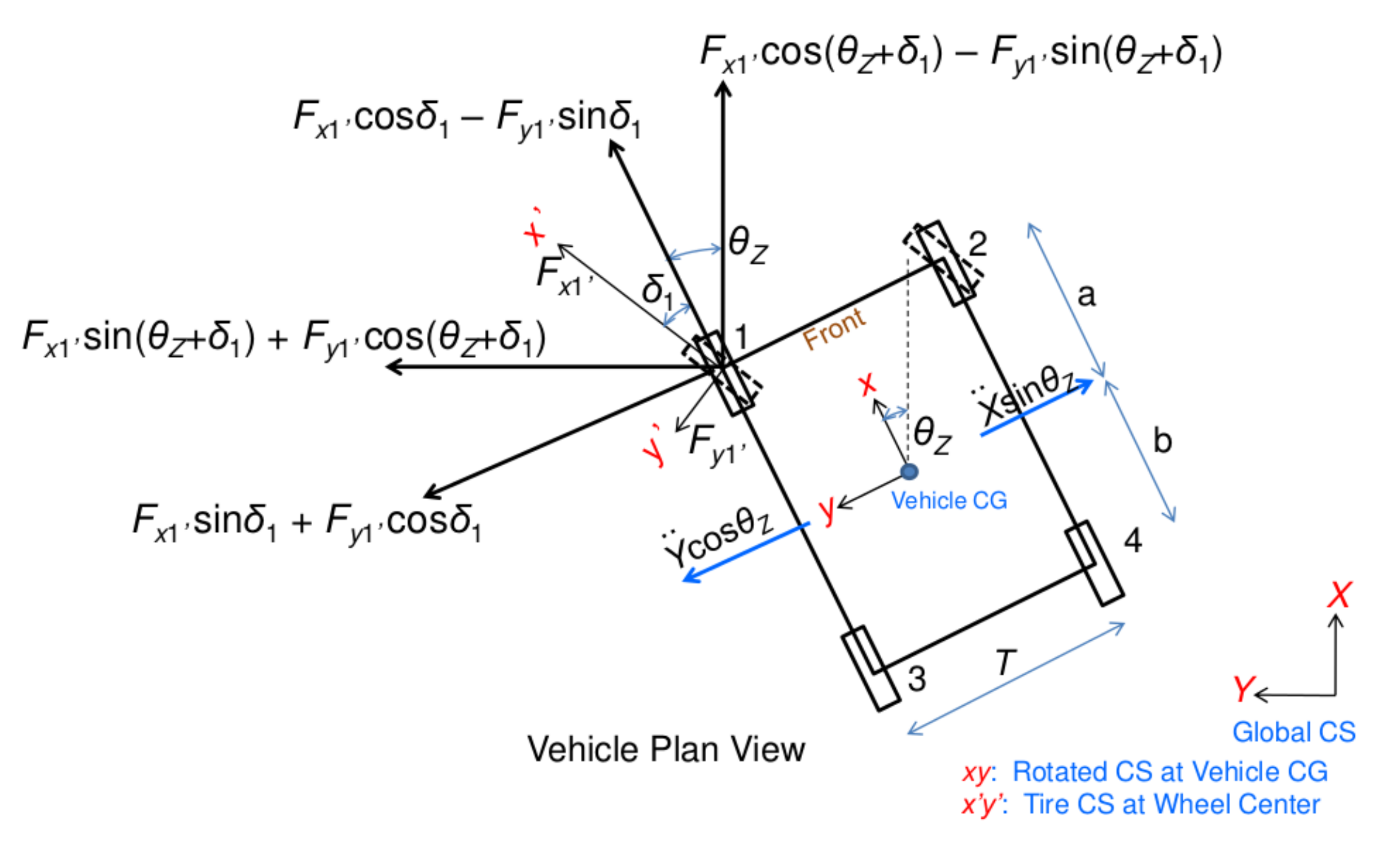}
\caption{Vehicle model (planar view)}\label{vehmod1}
 \centering
 \includegraphics[scale=0.34]{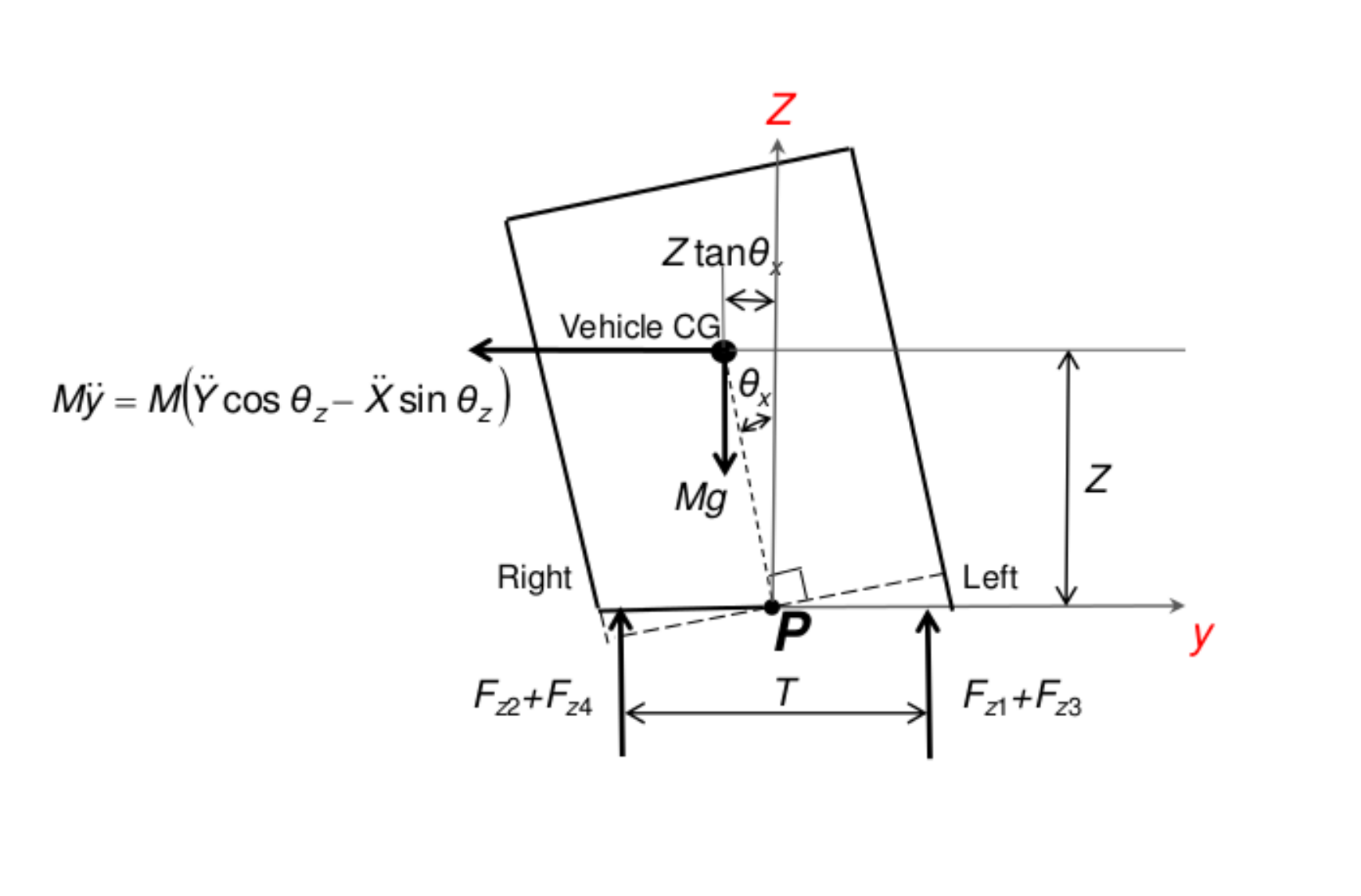}
\caption{Vehicle model (vertical view)}\label{vehmod2}
\end{figure}
The different parameters and constants (including those related to tire model), 
tire forces $F_{Yi}$ and wheel reaction forces $F_{Zi}$ $(i=1,\,2,\,3,\,4)$ 
are described in the following sub-section.
\subsection{Forces and parameters}
The wheel reaction forces $F_{Zi},\,i=1,\cdots,4$ are defined by the following formulae.
\begin{subequations}
\begin{eqnarray}
F_{Z1}&=&F_l+\frac{b}{2(a+b)}Mg+K\left\{Z_0-\left(Z+\frac{T}{2}\theta_X\right)\right\}-\nonumber \\ &&
        C\left(\dot Z+\frac{T}{2}\dot\theta_X\right)\\
F_{Z2}&=&F_r+\frac{b}{2(a+b)}Mg+K\left\{Z_0-\left(Z-\frac{T}{2}\theta_X\right)\right\}-\nonumber\\ &&
        C\left(\dot Z-\frac{T}{2}\dot\theta_X\right)\\
F_{Z3}&=&F_l+\frac{a}{2(a+b)}Mg+K\left\{Z_0-\left(Z+\frac{T}{2}\theta_X\right)\right\}-\nonumber\\ &&
        C\left(\dot Z+\frac{T}{2}\dot\theta_X\right)\\
F_{Z4}&=&F_r+\frac{a}{2(a+b)}Mg+K\left\{Z_0-\left(Z-\frac{T}{2}\theta_X\right)\right\}-\nonumber\\ &&
        C\left(\dot Z-\frac{T}{2}\dot\theta_X\right)
\end{eqnarray}
\end{subequations}
The tire forces $F_{Yi},\,i=1,\cdots,4$ are defined as 
\begin{equation} 
F_{Yi}=
\begin{cases}
0 & \text{if $F_{Zi}\leq 0$} \\
D\sin\left(C_T\arctan\left(B\phi\right)\right)\,&\text{otherwise}
\end{cases}
\end{equation}
where the parameters used in the definition of $F_{Yi}$ are calculated using the following formulae.
\begin{subequations}
\begin{eqnarray}
\phi&=&(1-E)(\alpha+\Delta S_h)\nonumber \\
           && +(E/B)\arctan(B(\alpha+\Delta S_h)) \\
D&=&a_1{F_{Zi}'}^2+a_2F_{Zi}'\\
B&=&\frac{a_3\sin\left(a_4\arctan\left(a_5F_{Zi}'\right)\right)}{C_TD}\\
E&=&a_6{F_{Zi}'}^2+a_7F_{Zi}'+a_8\\
\alpha&=&\frac{180}{\Pi}\bigg\{-\delta_i-  \nonumber \\ 
&& \arctan\bigg[\frac{\dot{X}\sin\theta_Z-\dot{Y}\cos\theta_Z-r_{Xi}\dot\theta_Z}
{\dot{X}\cos\theta_Z+\dot{Y}\sin\theta_Z-r_{Yi}\dot\theta_Z}\bigg]\bigg\}\\
F_{Zi}'&=&\frac{F_{Zi}}{1000}
\end{eqnarray}
\end{subequations}
A detailed description of the above tire model formulae is given in \cite{TyreModel} (also see \cite{pacejkabook}). 
The constant values used for numerical computations in this work  are given in the Appendix.
\subsection{Constraints on the Dynamics}\label{ConsSec}
\begin{enumerate}
\item\label{ConSTL} Suspension travel limits:
\begin{subequations}
\begin{eqnarray}
&& Z_{min}\leq Z+\frac{T}{2}\theta_X\leq Z_{max} ~ \mbox{and} \\
&& Z_{min}\leq Z-\frac{T}{2}\theta_X\leq Z_{max}. 
\end{eqnarray}
\end{subequations}
\item\label{ConFL} Controlling force limits: 
\begin{subequations}
\begin{eqnarray}
&& -F_{max}\leq F_l\leq F_{max}~ \mbox{and} \\
&& -F_{max}\leq F_r\leq F_{max}.
\end{eqnarray}
\end{subequations}
\item Anti-roll moment constraints (Inclusive Disjunctive Constraints): 
\begin{subequations}
\begin{eqnarray}
&& \bigg( -F_{Z1}-F_{Z3}\leq 0 \bigg)~ \bigvee ~ \nonumber \\
    &&  \bigg( \frac{\ddot Y \cos\theta_Z-\ddot X\sin\theta_Z}{g}-\frac{T}{2Z}\leq 0 \bigg) \label{EO1} \\
    &&  \mbox{and} \nonumber \\ 
&& \bigg( -F_{Z2}-F_{Z4}\leq 0  \bigg) ~ \bigvee ~ \nonumber \\
    && \bigg( -\frac{\ddot Y \cos\theta_Z-\ddot X\sin\theta_Z}{g}-\frac{T}{2Z}\leq 0 \bigg) \label{EO2}.
     \end{eqnarray}
\label{ConEO}     
\end{subequations}
\end{enumerate}
The disjunctive antil-roll constraints (\ref{ConEO}) are treated in the existing literature (e.g., see \cite{rajamaniBook, sharp,itspar}) 
in the more conservative (conjunctive) form:
\begin{equation}
-F_{Z1}-F_{Z3}\leq 0~ \mbox{and} -F_{Z2}-F_{Z4}\leq 0 \label{ConCons}
\end{equation}
which imply that the wheels do not lift off the ground. This treatment of constraints fails 
to generate any controller intervention when the wheels have lifted off.
Further, it also does not check whether the vehicle roll has become unstable 
(see rollover index analysis in section \ref{sec:ri}) and needs controller intervention at all. 
\section{Computation of the Control Forces}\label{sec:Numerical}
The switching of the disjunctive constraints, as expected, introduces stiffness in the dynamics resulting in 
a possibly large condition number in the constraint Jacobian of the NLP solver.
The $\alpha$-method when used as the discretization method in a transcription
scheme can affect regularization of the Jacobian (\cite{GAlpha}). 
Discretizations that do not affect a regularization, such as the backward Euler discretization, 
fail to converge to a solution in the NLP problem.
\par
Hence, the optimal control problem with its dynamics and disjunctive constraints 
as described in section \ref{MathModel} is numerically solved by the 
direct transcription method (cf. section \ref{SecTrans}) with $\alpha$-method discretization 
over the maneuver simulation interval $[t_0, t_f]$, which is partitioned into $N$ equally spaced grid points. 
The step size for the $\alpha$-method is then $h=\frac{t_f-t_0}{N-1}$. 
Various combinations of the constraints and of the initial guesses for the control profile are experimented with for 
obtaining the numerical solutions. The {\tt fmincon} \cite{Gilbert:2006:numerics} function of the 
MATLAB$\textsuperscript\textregistered$, a sequential quadratic 
programming  (SQP) routine, is used as an NLP solver. It is seen that  the $\alpha$-method discretization with $N\geq 121$ uniform 
time steps captures the required time resolution of the stiff dynamics of the maneuvers and produces convergent numerical solutions. 
\subsection{Numerical Solutions with Disjunctive Constraints} \label{subsec:NumEO} 
 \begin{figure}
\centering
\includegraphics[scale=0.18]{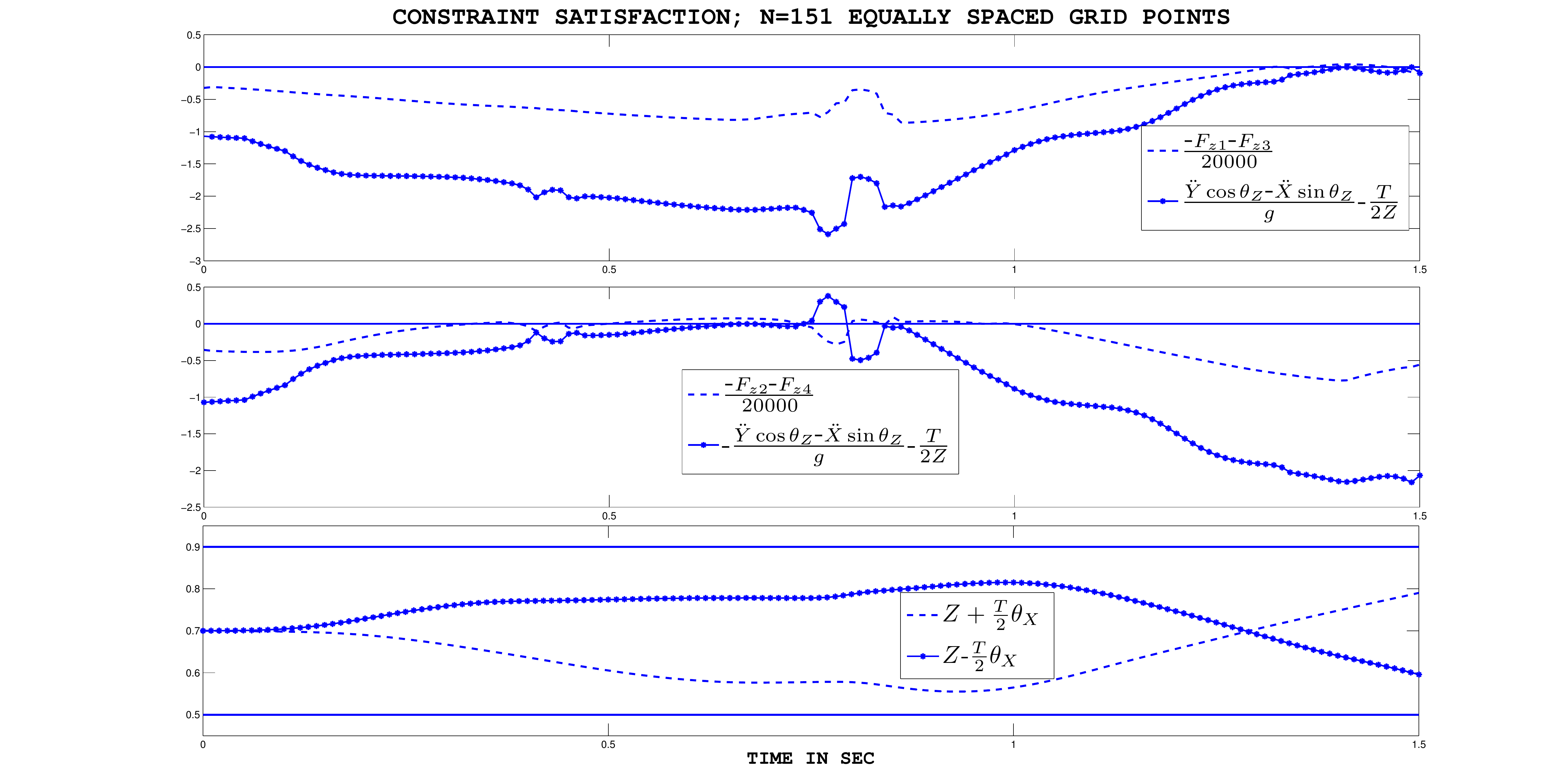}
\caption{Satisfaction of Disjunctive Constraints; N=151}\label{EO_MPE_IndFLFrCG1N151ConstraintFig}
\end{figure}
\begin{figure}
 \centering
 \includegraphics[scale=0.18]{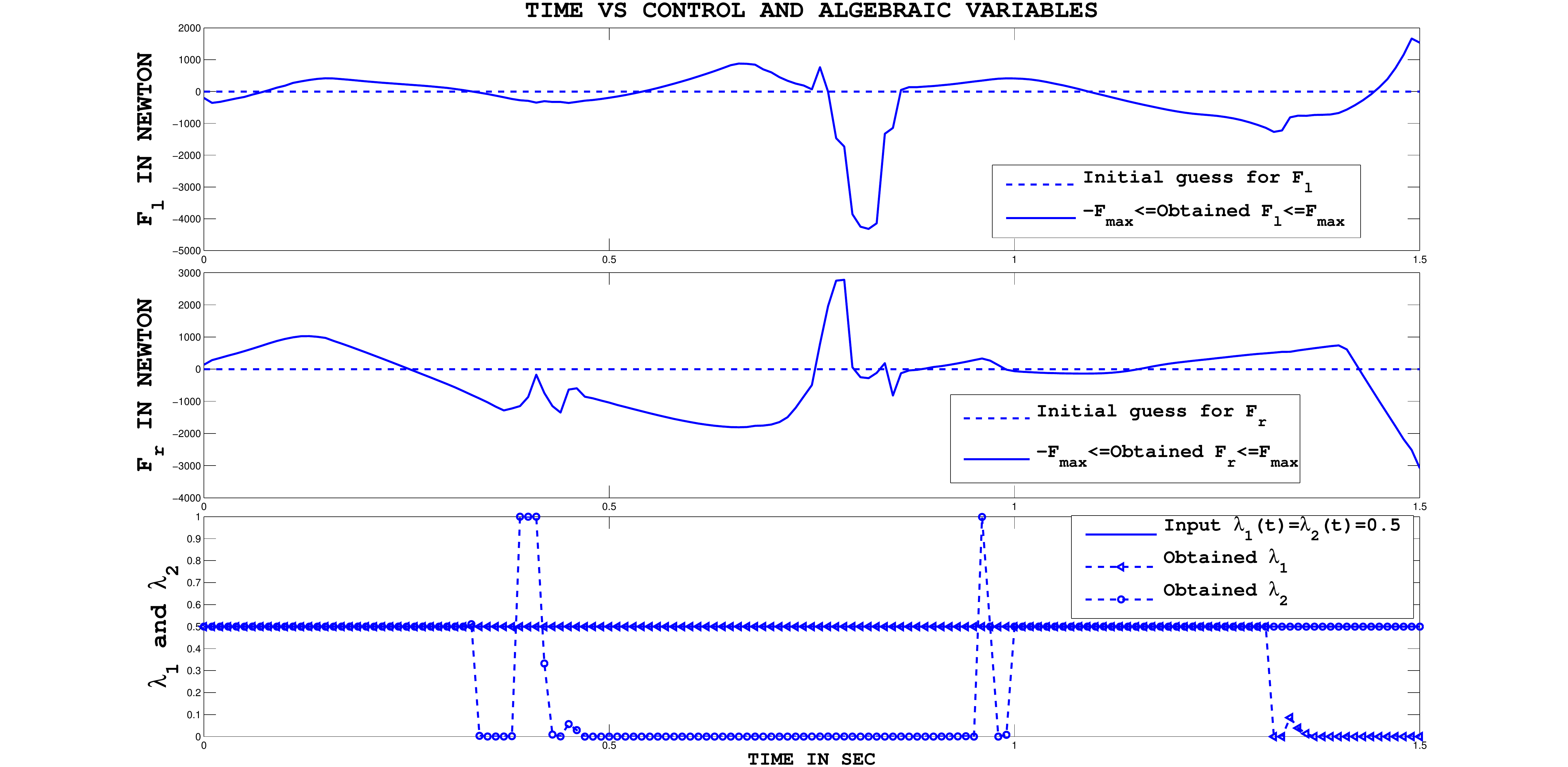}
\caption{Control and algebraic variables with Disjunctive Constraints and initial guess value of zero; $N=151$}\label{EO_MPE_IndFLFrCG1N151CtrlAlgFig}
\end{figure}
\begin{figure}
\centering
 \includegraphics[scale=0.18]{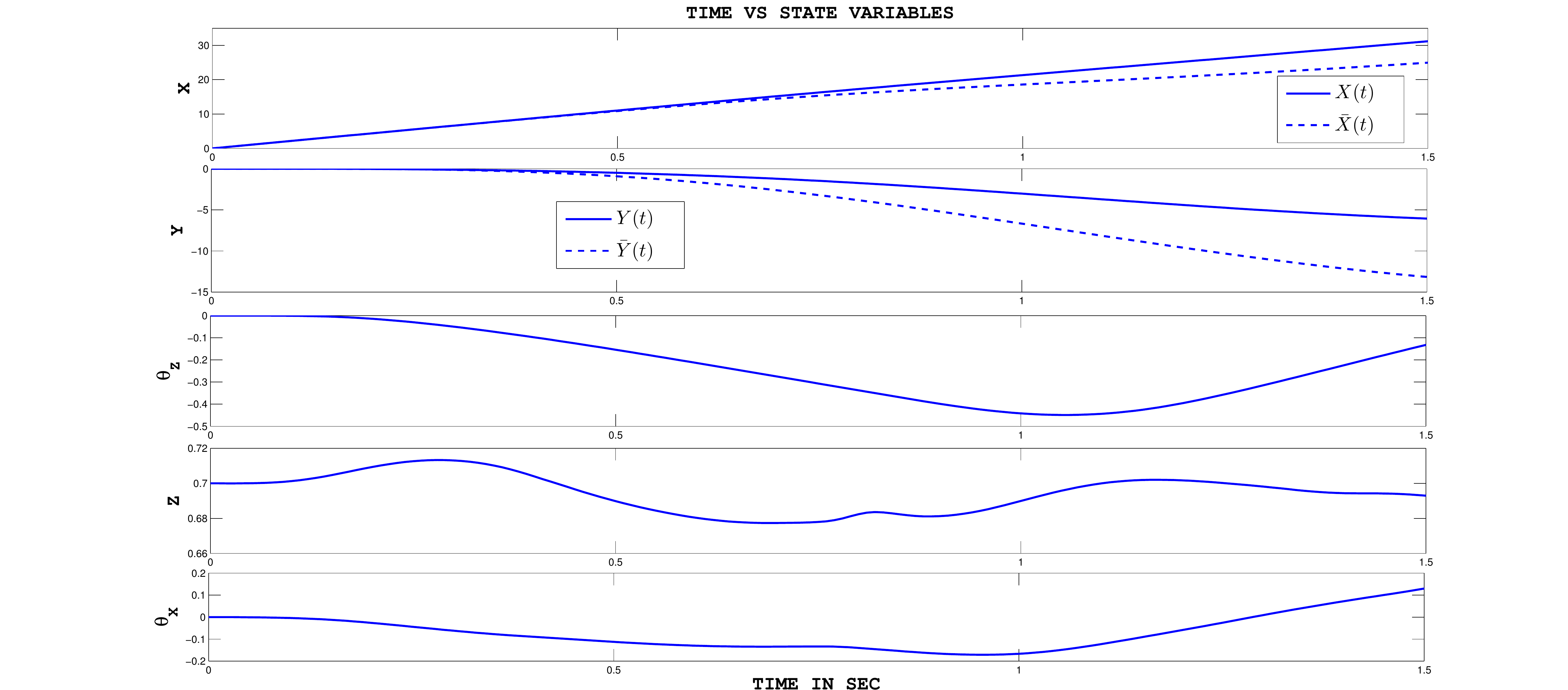}
\caption{State variables with Disjunctive Constraints; $N=151$}\label{EO_MPE_IndFLFrCG1N151StateFig}
 \end{figure}
A severe fishhook maneuver with steering input as in figure \ref{Fig:SteeringFun} and disjunctive constraints (i.e., wheels lifting off) are
 employed. Figures \ref{EO_MPE_IndFLFrCG1N151ConstraintFig}, 
\ref{EO_MPE_IndFLFrCG1N151CtrlAlgFig} and 
 \ref{EO_MPE_IndFLFrCG1N151StateFig} show the controls obtained from the simulation of the 
disjunctively constrained dynamic optimization problem with a partition of $N=151$ grid points 
 with the initial guess values of $F_l$ and $F_r$ being zero (i.e., inactive).
 Figure \ref{EO_MPE_IndFLFrCG1N151CtrlAlgFig} shows the evolution of the control forces needed to
 stabilize the vehicle. Using initial guess of control variables for direct transcription as 
$F_l(t)=F_r(t)=1000$ instead of zeros and a uniform grid with $N=121$ points, the dynamic 
optimization yields the results shown in 
figures \ref{EO_MPE_IndFLFrCG3N121ConstraintFig} and \ref{EO_MPE_IndFLFrCG3N121CtrlAlgFig}.  
Figure \ref{EO_MPE_IndFLFrCG3N121ConstraintFig} shows that the inclusive disjunctive constraints are satisfied. 
The time profile of the variables $X,\,Y,\,\theta_Z$ is similar to that
in figure \ref{EO_MPE_IndFLFrCG1N151StateFig} while $Z$ and $\theta_X$ shows minor differences with the zero initial guess.
However, the obtained control forces are very different from the ones in figure \ref{EO_MPE_IndFLFrCG1N151CtrlAlgFig} 
implying that there are possibly several local optima.
 \begin{figure} 
\centering
\includegraphics[scale=0.185]{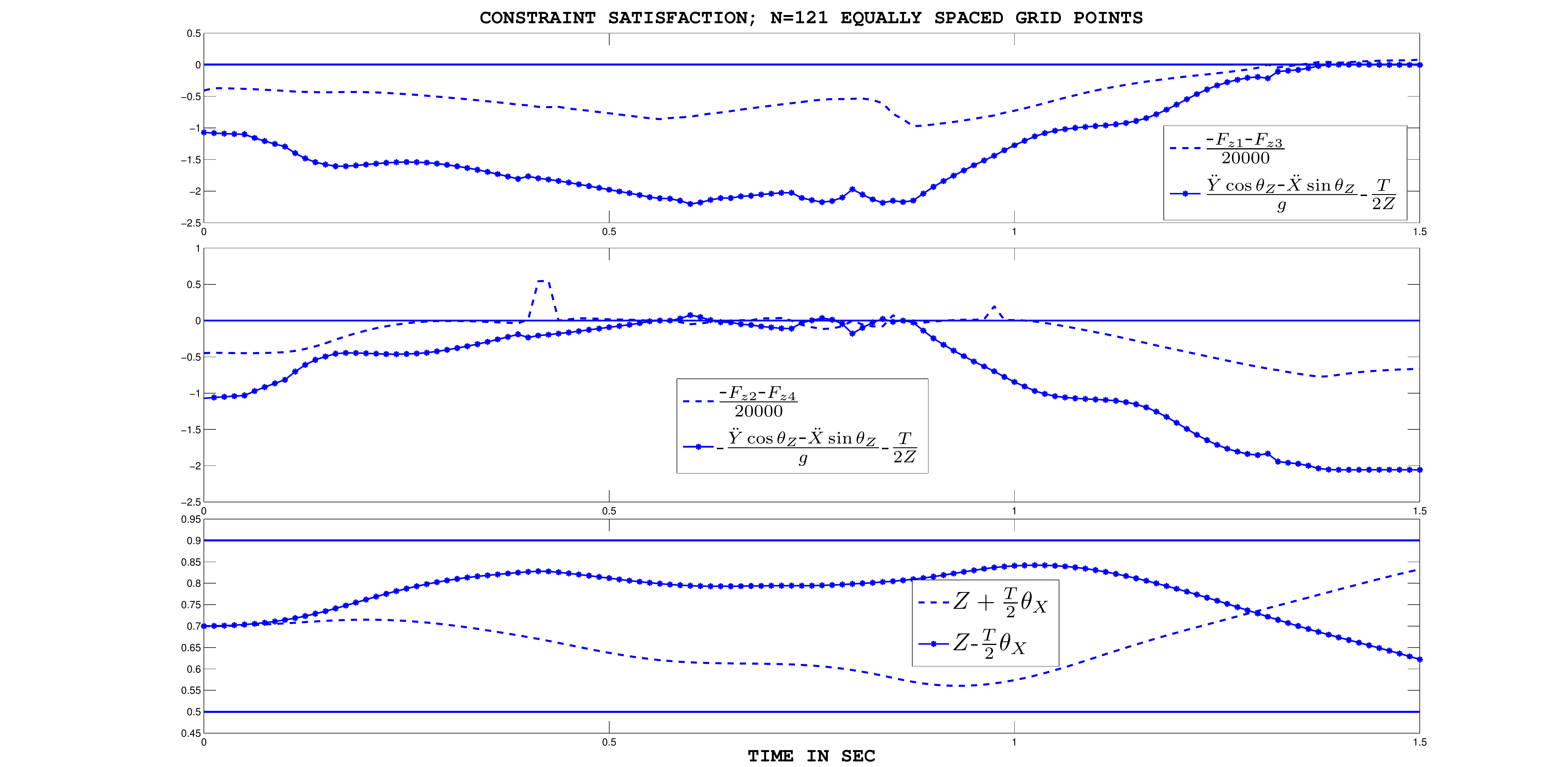}
\caption{Satisfaction of Disjunctive Constraints; $N=121$}\label{EO_MPE_IndFLFrCG3N121ConstraintFig}
\end{figure}
\begin{figure}
\centering
\includegraphics[scale=0.2]{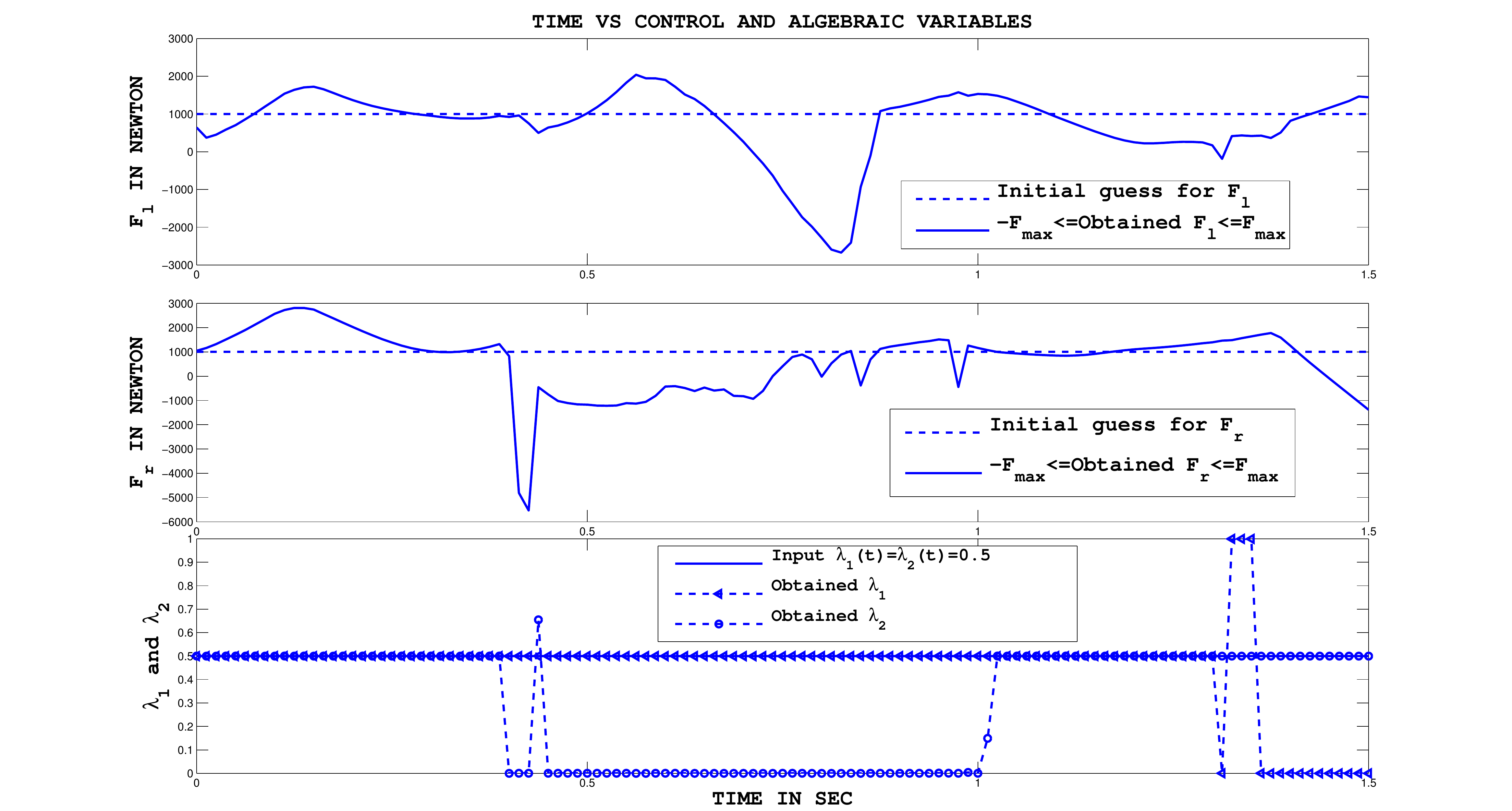}
\caption{Control and algebraic variables with Disjunctive Constraints; $N=121$}
\label{EO_MPE_IndFLFrCG3N121CtrlAlgFig}
\end{figure}
However, for the realization on a active suspension system, it is desirable to have a single time profile for the control 
forces insensitive to the perturbation of initial guesses and grid size in the optimizer.
This can be accomplished using a set of anti-symmetric control forces, i.e., $F_l = -F_r$. 
With this additional constraint $F_l + F_r = 0$, the solutions are obtained with different initial guesses as before. 
The controls then become almost insensitive to the perturbation of initial guesses and of grid sizes, 
as the anti-symmetric control force constraint resolves the problem of the local optima.
\begin{figure} 
 \centering
 \includegraphics[scale=0.2]{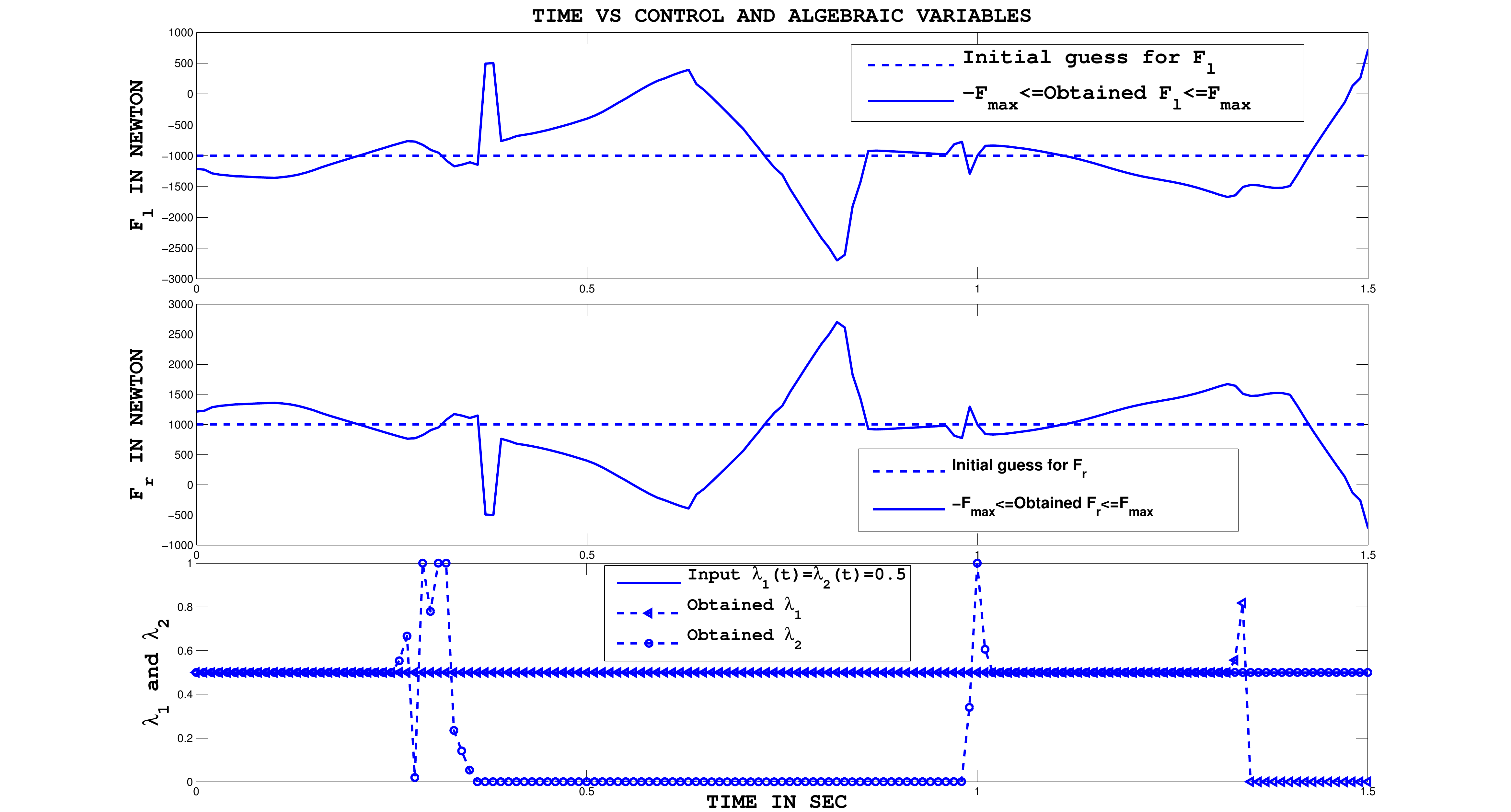}
 \caption{Anti-Symmetric Controls with Disjunctive Constraints; $N=151$}
 \label{EO_MPE_IndFLFrZrCG3N151CtrlAlgFig}
\end{figure}
\begin{figure}
 \centering 
 \includegraphics[scale=0.18]{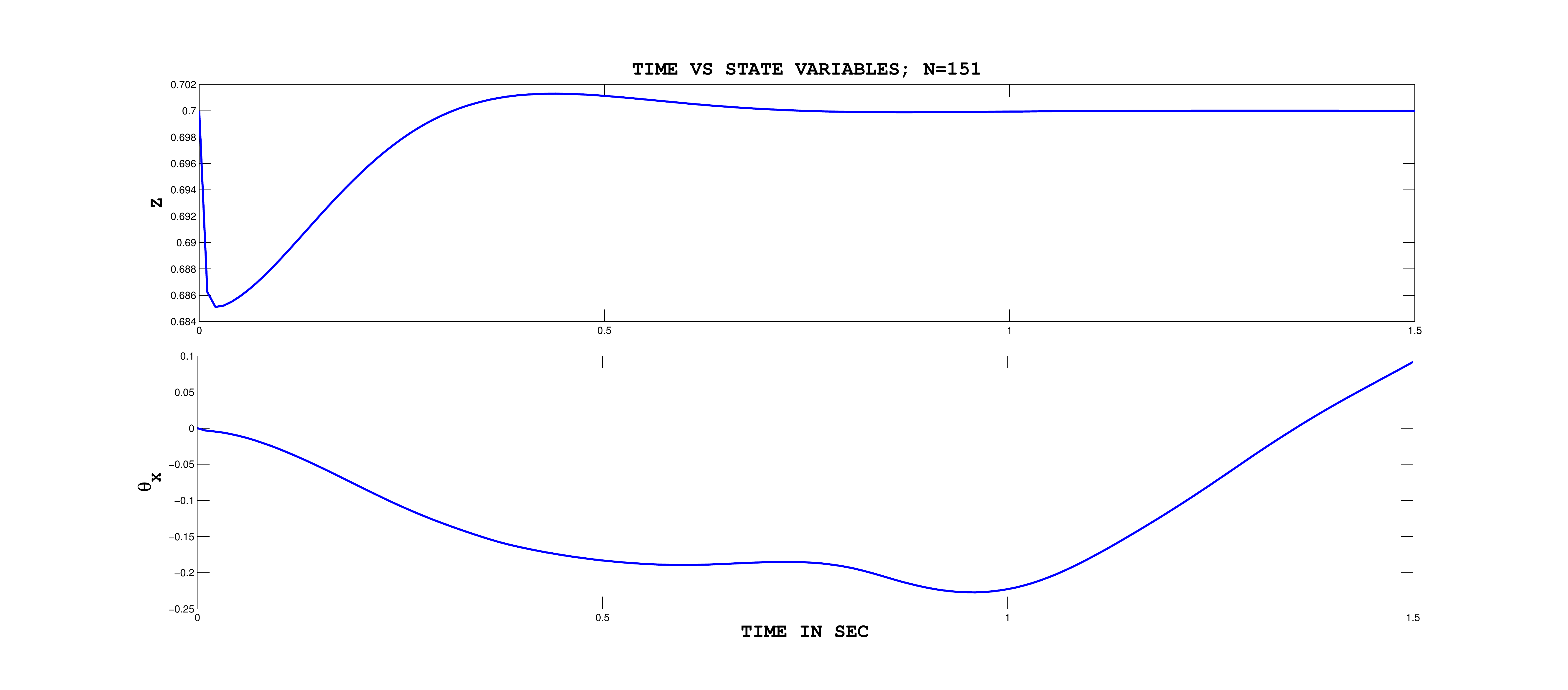}
 \caption{State Variables for Anti-Symmetric Controls with Disjunctive Constraints}
 \label{EO_MPE_IndFLFrZrCG3N151StateFig}
 \end{figure}
Figure \ref{EO_MPE_IndFLFrZrCG3N151CtrlAlgFig} shows the control forces, satisfying the anti-symmetric force constraints and computed with 
the initial guess $F_r(t) = -F_l(t) = 1000$ using $N=151$ grid points, while
figure \ref{EO_MPE_IndFLFrZrCG3N151StateFig} shows the time profile of $Z$ and $\theta_X$.
\section{Rollover Index and Efficacy of Disjunctive Constraints}\label{sec:ri}
To underscore the effectiveness of the disjunctive constraints approach over the conservative constraints, we use the concept 
of rollover index \cite{rajamaniBook}. The rollover index  $R$ is defined as 
\begin{equation*}
R=\frac{(F_{Z2}+F_{Z4})-(F_{Z1}+F_{Z3})}{(F_{Z2}+F_{Z4})+(F_{Z1}+F_{Z3})}.
\end{equation*}
The wheel lift-off starts occurring at $|R|=1$. $|R|<1$ indicates no lift-off of the wheels and hence no rollover. 
On the other hand $|R|>1$ indicates lift-off of the wheel. But not all 
wheel lift-off causes rollover. This is exactly where the present
disjunctive constraint (\ref{ConEO}) approach becomes useful in computing anti-rollover forces
in the suspensions, i.e., a correct negative roll moment can stabilize the vehicle even when the wheels have lifted off
the ground. 
The following comparison of absolute values of rollover index over the 
simulation interval shows that the disjunctive approach is a more realistic approach.  \par
Let us consider the vehicle control problem with the following computation for both the disjunctive constraint
and the conventional conservative approaches. 
The transcription is 
done with $N=151$ equally spaced grid points along with the initial guess for controls as $F_l=F_r=0$. 
The solution plots for the disjunctive constraints are given in figure \ref{EO_MPE_IndFLFrCG1N151CtrlAlgFig}; 
these indicate the absence of rollover. 
An anti-rollover solution with conservative constraints (no wheels are allowed to lift-off) is computed, as shown in
figure \ref{Fig:indConsCtrlZero}. 
 \begin{figure} 
 \includegraphics[scale=0.15]{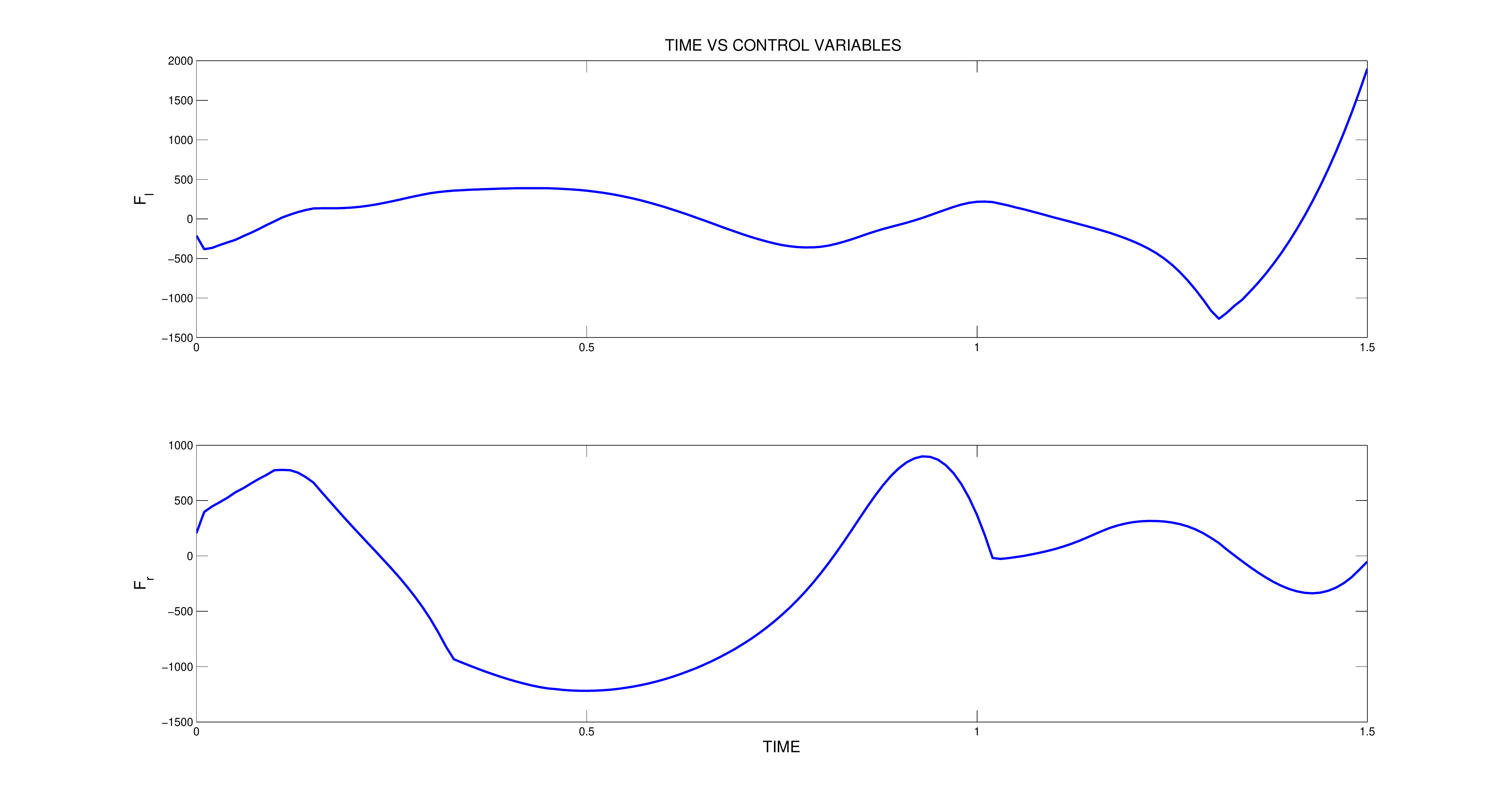}
\caption{Control variables $F_l,\,F_r$ computed using the conservative constraint approach with initial guess value of zero; $N=151$}\label{Fig:CON_MPE_IndFLFrCG1N151Ctrl}
\label{Fig:indConsCtrlZero}
\end{figure} 
Rollover index from both the simulations are plotted in figure \ref{Fig:Rolloverindex}. 
Figure \ref{Fig:Rolloverindex} shows that the disjunctive constraints allow the wheels to lift-off
and yet the vehicle gets stabilized by the computed controls. 
The same figure compares the disjunctive constraint approach to the conventional conservative approach 
which does not cover the more severe situation of lift off
of the wheels. Thus, whenever there is room for stabilization and prevention of rollover 
even when the wheels have lifted off during a severe maneuver, the disjunctive constraint approach provides an 
effective way to compute the control forces in active suspensions.
\par
Figure \ref{Fig:ri2} shows the rollover index obtained from the computation by the transcription method to find anti-symmetric 
controls (see figure \ref{EO_MPE_IndFLFrZrCG3N151CtrlAlgFig}) for the disjunctive dynamics. In this case too, the vehicle is
stabilized with anti-roll moment induced by the anti-symmetric control forces at the suspensions even though wheels are allowed
to be lifted off (corresponding to $|R| > 1$). 
\begin{figure}
\centering
\includegraphics[scale=0.193]{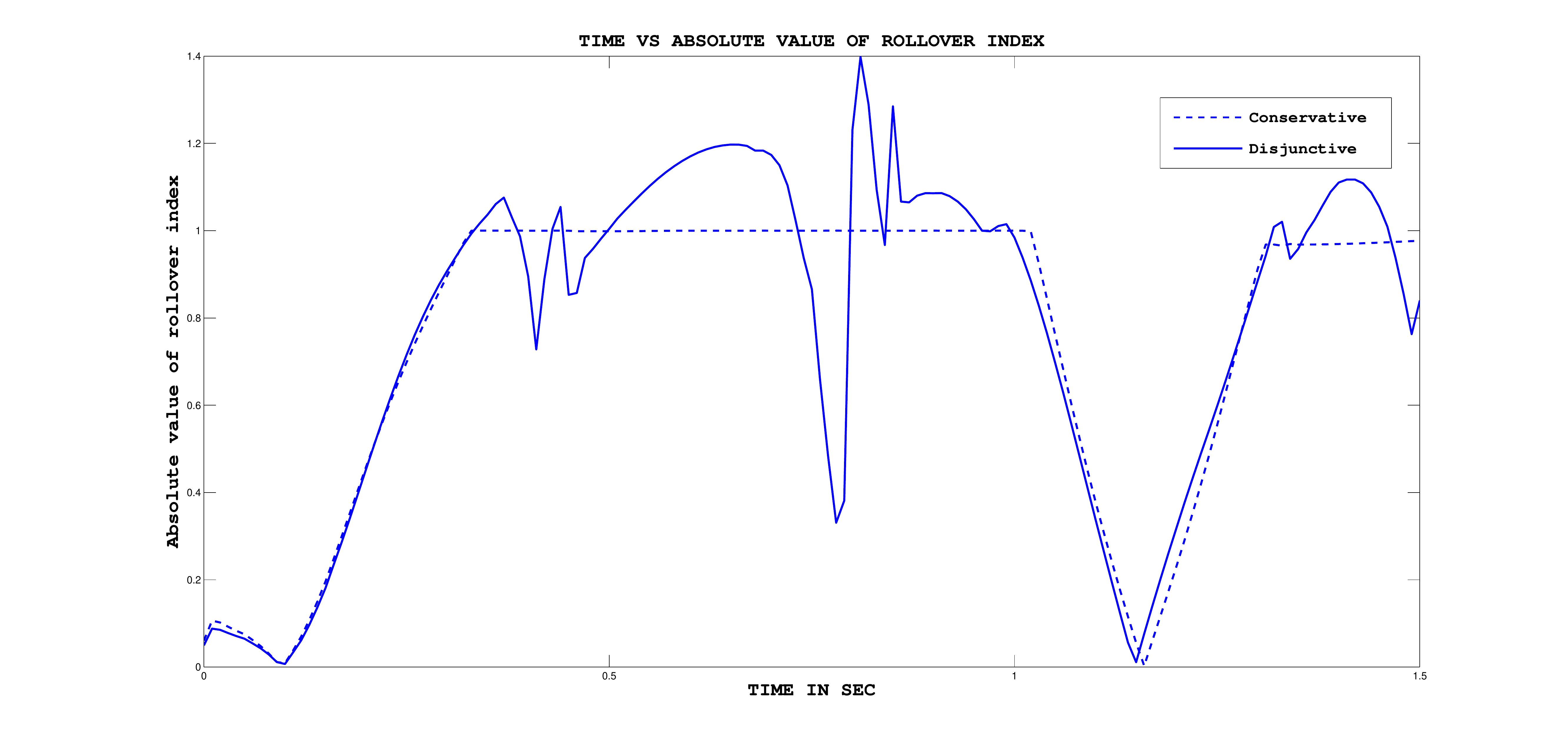}
\caption{Comparison of the magnitude of the rollover index over time. The present model using disjunctive constraints show stabilization even
after wheels have lifted off (rollover index magnitude greater than one) while the existing conservative approach does not allow wheels to lift off.}\label{Fig:Rolloverindex}
\end{figure}
\begin{figure}
\centering
\includegraphics[scale=0.193]{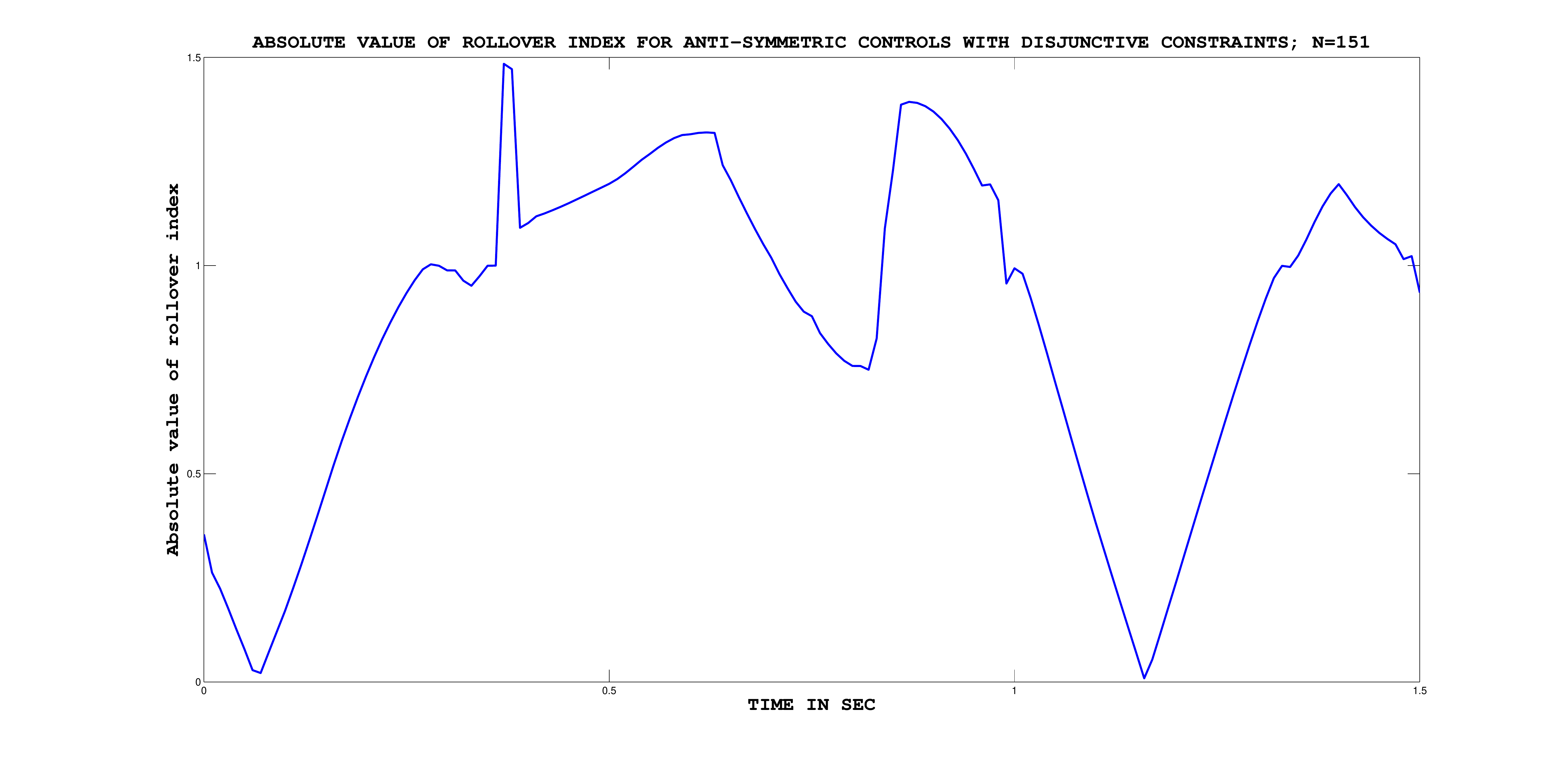}
\caption{Magnitude of the rollover index over time with anti-symmetric control forces and disjunctive constraints. The present model using disjunctive constraints show stabilization even
after wheels have lifted off (rollover index magnitude greater than one).}\label{Fig:ri2}
\end{figure}
\section{Synthesis of Control Forces in terms of Sensor Data Output} \label{synth}
In order to develop an effective control system, it is necessary to investigate if the control forces can be represented as a linear 
combination of sensible parameters so that the sensor output data from the system can be used to synthesize the control forces 
in the active suspensions. 
Specifically, we seek to determine the 
(local) optimum values (in the neighborhood some values that are useful and attainable from the engineering
point of view) of the coefficients $\varphi$'s in the formula 
$$F_l=\varphi_1\theta_X+\varphi_2\dot\theta_X+\varphi_3\dot\theta_Z+
\varphi_4\left(Z-Z_0\right)+\varphi_5\dot Z$$ along with the constraint $F_r=-F_l$.  
For synthesis of the control force, we find the dominating terms in the above formula to determine which sensor parameters
are critical and determine the force. The linear combination of the dominant terms which best approximates $F_l$ is determined. 
The control output is synthesized from these terms weighted by the respective coefficients. The weights computed for 
a range of maneuvers similar to that in our numerical computation are stored
in a look up table on the embedded computing element of the controller and fed forward into the active suspension.
Using $\hat\varphi$ to denote the vector 
$\begin{bmatrix}
\varphi_1 & \varphi_2 & \varphi_3 & \varphi_4 & \varphi_5
\end{bmatrix}$, 
a local optimum is obtained with the initial guess $\hat\varphi=0$ 
and $N=151$ grid points to yield $\hat\varphi=
\begin{bmatrix}
-916.5607 & -2102.4 & -4799.4 & 3.8244\times 10^{-4} &  -0.0078
\end{bmatrix}$. 
The time-variation of $F_l$ is shown in figure \ref{EO_MPE_FlFrZrLCCG1N151_CtrlAlgFig}. 
Also, the magnitudes of the individual terms that make up the equation 
$F_l=\varphi_1\theta_X+\varphi_2\dot\theta_X+\varphi_3\dot\theta_Z+
\varphi_4\left(Z-Z_0\right)+\varphi_5\dot Z$ is shown in figure \ref{EO_MPE_FlFrZrLCCG1N151_CtrlAlgFig}. 
It is apparent that the term $\varphi_3\dot\theta_Z$ is the most dominating term and 
approximates the total control force the best. This means, applying a reaction force proportionate to the rate of yaw 
stabilizes the roll over tendency of the vehicle undergoing fishhook maneuver with wheels lifting off the ground.
\begin{figure} 
 \centering
\includegraphics[scale=0.182]{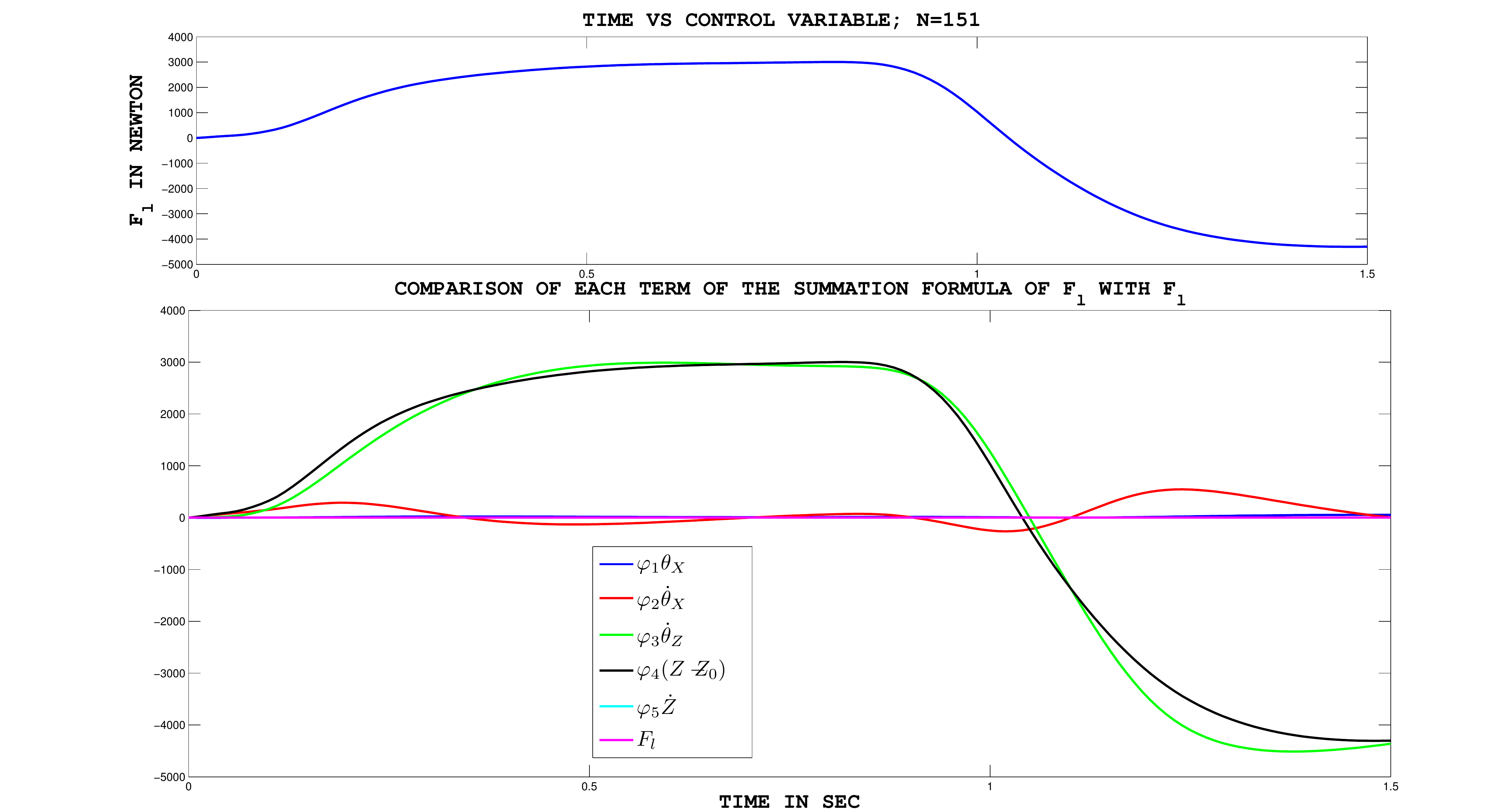}
\caption{Control variable $F_l$ for Disjunctive Constraints along with its Approximation}
\label{EO_MPE_FlFrZrLCCG1N151_CtrlAlgFig}
\end{figure}
In another numerical experiment with different initial guess (active forces) we obtain $F_l$ as shown in 
figure \ref{Fig:EOFinalCtrl} and this closely matches the one in figure \ref{EO_MPE_FlFrZrLCCG1N151_CtrlAlgFig},
showing the robustness of the approximation, i.e., insensitivity to perturbation of initial guesses in the optimizer. 
The coefficient $\varphi_3$ is found in this case to be $-4796.2$ which is comparable to its value from the previous
example, with an initial guess that the control forces are inactive. This shows an 
effective way to synthesize the control forces that can vary linearly in proportion to the 
sensed parameter $\dot{\theta}_Z$ and stabilize the vehicle in spite of the wheels lifting off the ground.  
\begin{figure} 
\centering
\includegraphics[scale=0.181]{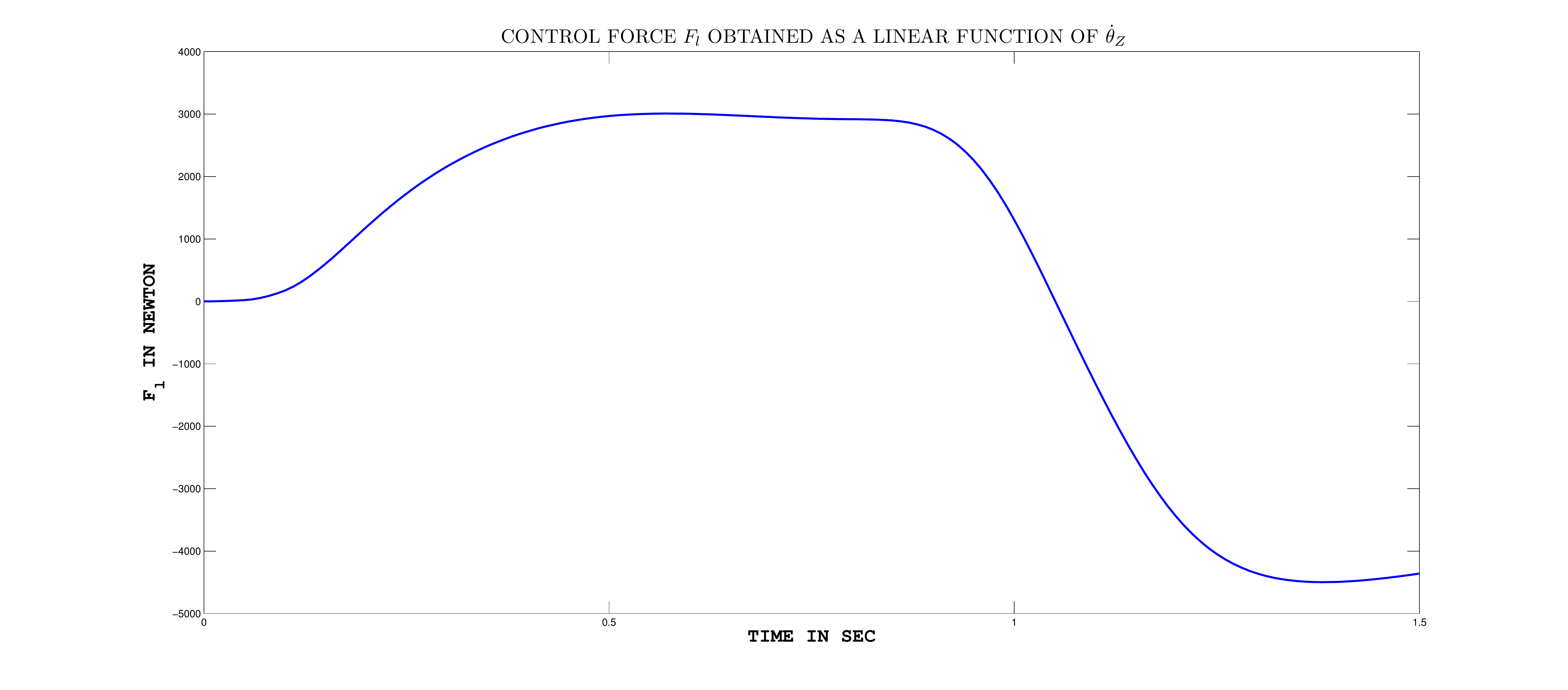}
\caption{Control variable $F_l$ expressed as $F_l=\varphi_3\dot\theta_Z$ for Disjunctive Constraints}
\label{Fig:EOFinalCtrl}
\end{figure}
\subsection{Disjunctive Constraints vis-a-vis Conservative Constraints}\label{subsec:NumConserv}
We compare solutions with disjunctive constraints with those corresponding to the conservative constraints (\ref{ConCons}) 
defined in section \ref{ConsSec}. 
Choosing the same initial guesses and the step-size as for the solution shown in 
figure \ref{EO_MPE_IndFLFrZrCG3N151CtrlAlgFig}, the control forces obtained for the conservative case 
are shown in figure \ref{Fig:CON_MPE_IndFLFrCG3N151Ctrl}.
\begin{figure} 
 \centering
\includegraphics[scale=0.2]{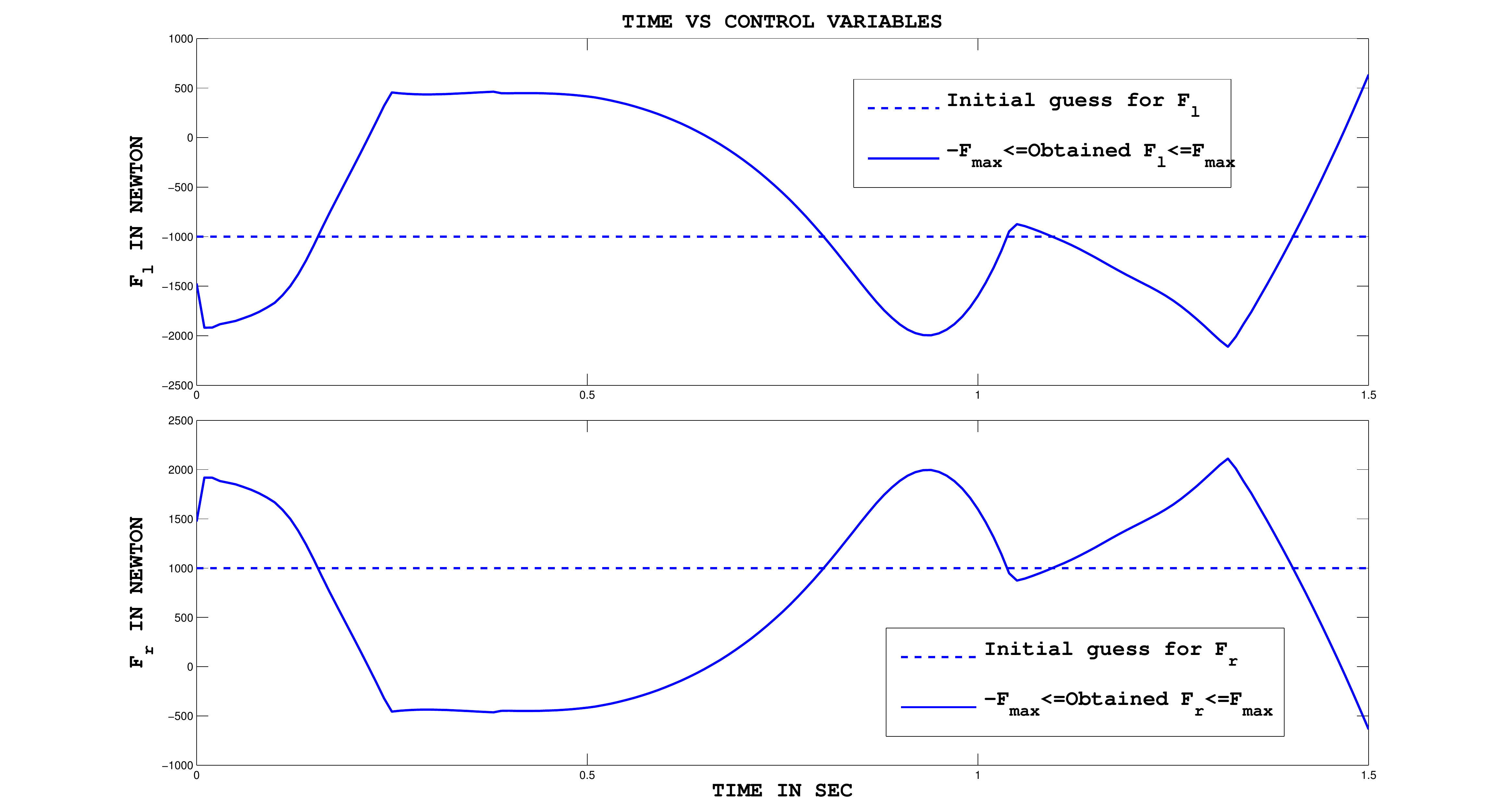}
\caption{Anti-Symmetric Controls with Conservative Constraints; $N=151$}\label{Fig:CON_MPE_IndFLFrCG3N151Ctrl}
\end{figure}
As before we model the control forces as a 
linear combination of sensible parameters
and find the (local) optimum values of the coefficients $\varphi$'s in the formula $F_l=\varphi_1\theta_X+\varphi_2\dot\theta_X+\varphi_3\dot\theta_Z+
\varphi_4\left(Z-Z_0\right)+\varphi_5\dot Z$ along with the constraint $F_r=-F_l$. 
Figure \ref{Fig:CON_MPE_FlFrZrLCCG1N151_Ctrl} shows the control forces $F_l$ and $F_r$ 
obtained with the initial guess $\varphi_i=0,$ $\forall i=1,\cdots,5$ and with $N=151$ grid points, while figure \ref{Fig:CON_MPE_FlFrZrLCCG1N151_Constraint}
shows the satisfaction of the conservative constraints. The (local) optimal value of the coefficients $\varphi_i$'s in the above 
linear combination formula of $F_l$ is found to be 
\begin{equation*} \hat\varphi= 
\begin{bmatrix}
-800.8541 & -1556.5 & -4763.9 &  0 &   6.0974 \times 10^{-4}
\end{bmatrix}. \end{equation*} 
In figure \ref{Fig:CON_Dom_MPE_FlPlFrZrLCCG1N151_State} we find 
that $\varphi_3\dot\theta_Z$ is the dominating term approximating $F_l$ the closest. 
\begin{figure} 
\centering
\includegraphics[scale=0.186]{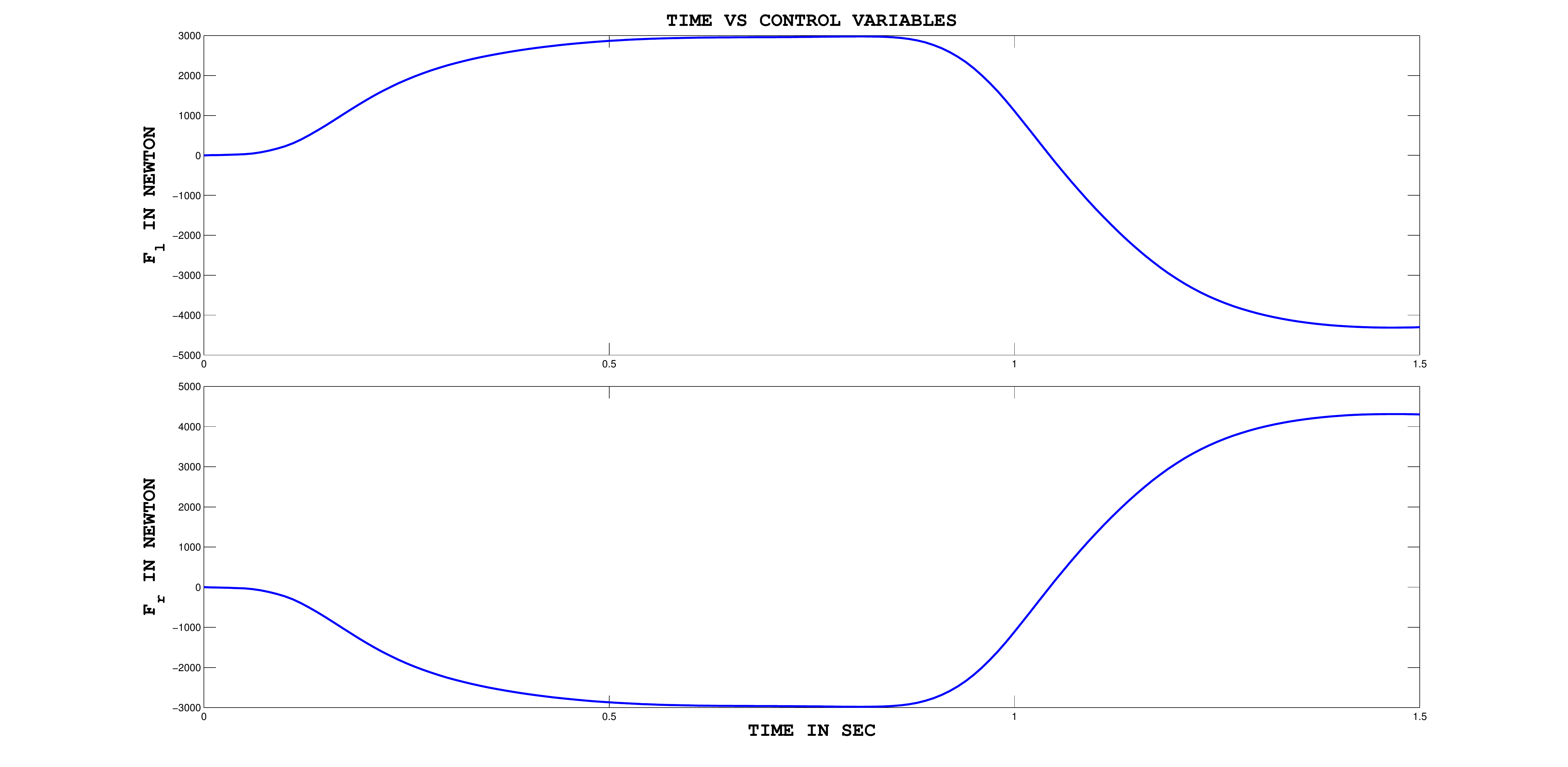}
\caption{Controls  $F_l,\,F_r$ computed with Conservative Constraints using initial guess of zero; $N=151$}
\label{Fig:CON_MPE_FlFrZrLCCG1N151_Ctrl}
\end{figure}
\begin{figure} 
\centering
\includegraphics[scale=0.186]{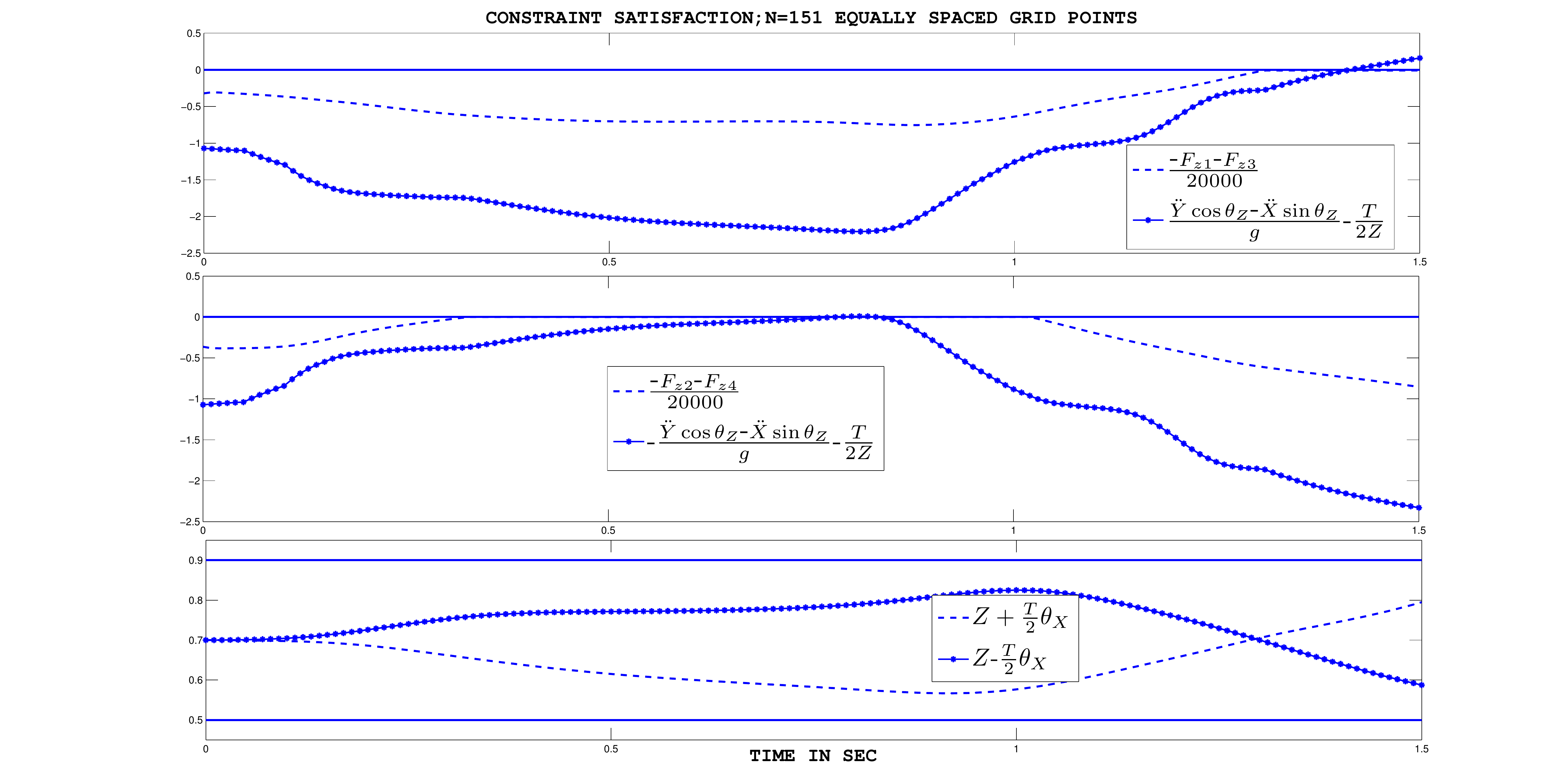}
\caption{Satisfaction of Conservative Constraints with Initial Guess of Zero Control Forces; $N=151$}
\label{Fig:CON_MPE_FlFrZrLCCG1N151_Constraint}
\end{figure}
Hence, as before, $F_l$ is approximated by $\varphi_3\dot\theta_Z$ alone. 
The control force $F_l$ is re-calculated setting $F_l =\varphi_3\dot\theta_Z$  
and $\varphi_3=-4761.2$ is obtained. This is shown in figure \ref{Fig:CON_Final_Ctrl}.  
\begin{figure}
\centering
\includegraphics[scale=0.15]{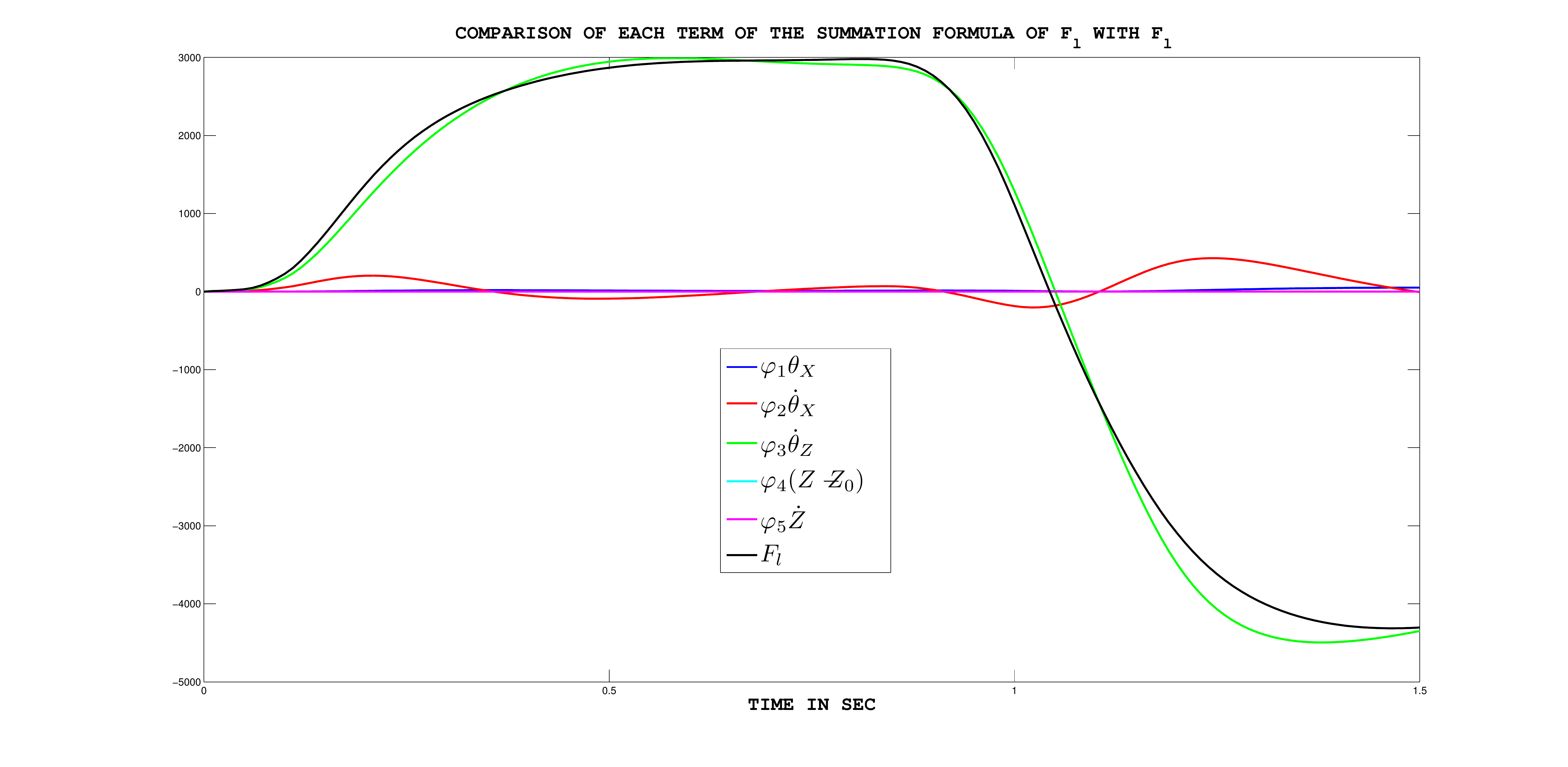}
\caption{Approximation of Control Force $F_l$, for Anti-Symmetric Controls, by each term for Conservative Constraints.}
\label{Fig:CON_Dom_MPE_FlPlFrZrLCCG1N151_State}
\includegraphics[scale=0.15]{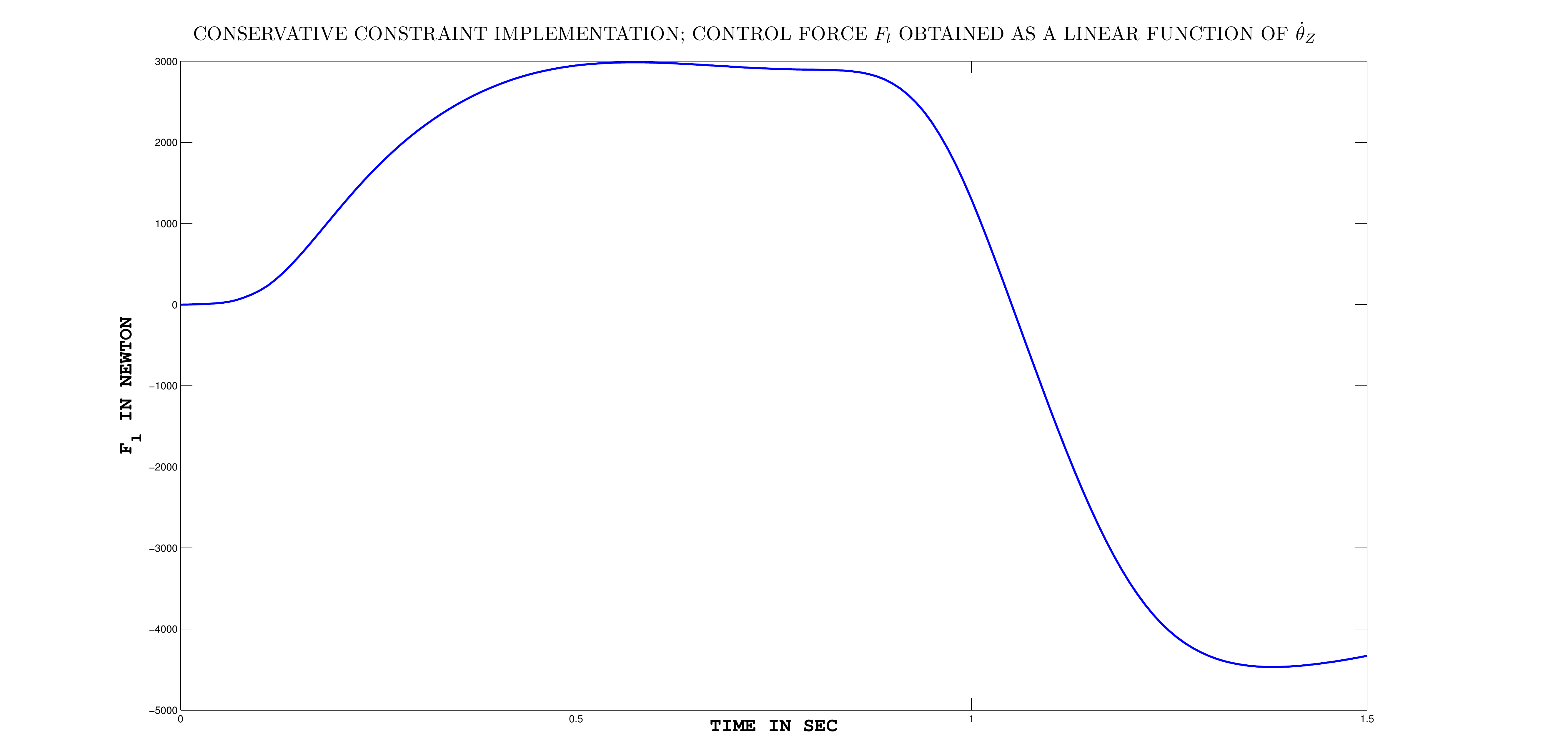}
\caption{Control force $F_l$ expressed as $F_l=\varphi_3\dot\theta_Z$ 
with Conservative Constraints for Anti-Symmetric Control Forces.}
\label{Fig:CON_Final_Ctrl}
\end{figure} 

It is apparent that with conservative constraints, the force requirements are comparable to
that when the inclusive disjunctive constraints are used, although the latter ones correctly stabilize the roll in spite of
the wheels being lifted off.
Thus, control forces as linear functions of  $\dot\theta_Z$ (sensed yaw rate data) 
while satisfying the disjunctive constraints (i.e., with wheels lifting off) is realizable within the existing hardware systems but 
would ensure safety under more severe maneuvering conditions.
\section{Validation using Arbitrary Steering Maneuver with the Synthesized Controls as Input}
In the preceding sections we found that in case of the anti-symmetric controls, 
$F_l$ (and hence $F_r$) can be taken as a linear function of 
the yaw rate, i.e.,  $\dot\theta_Z$. The same synthesized controls are then used against an arbitrary steering input given in 
figure \ref{Fig:ArbSteeringFun} (cf. \cite{ramani}) and we check whether  the disjunctive constraints are satisfied.
\begin{enumerate} 
\item{Inclusive disjunctive constraints, anti-symmetric controls:} 
$F_l$ is approximated by the linear formula $F_l=\varphi_3\dot{\theta}_Z$ and
$\varphi_3=-4796.2$ as computed in the previous section. 
The simulation of the dynamics with $F_l=-4796.2\dot{\theta}_Z$ 
and $F_r=-F_l$ yields the plots given in figure \ref{Fig:CON_EO_Arb_Steer}, which in turn shows that the disjunctive 
constraints are satisfied.
\item{Conservative constraints, anti-symmetric controls:} The synthesized control 
in this case is $F_l=\varphi_3\dot\theta_Z$, $\varphi_3=-4761.2$. The plots in the same figure \ref{Fig:CON_EO_Arb_Steer} verify that 
the constraints are satisfied when the synthesized force is used as the control input.
\end{enumerate}
\begin{figure}
\centering
\includegraphics[scale=0.186]{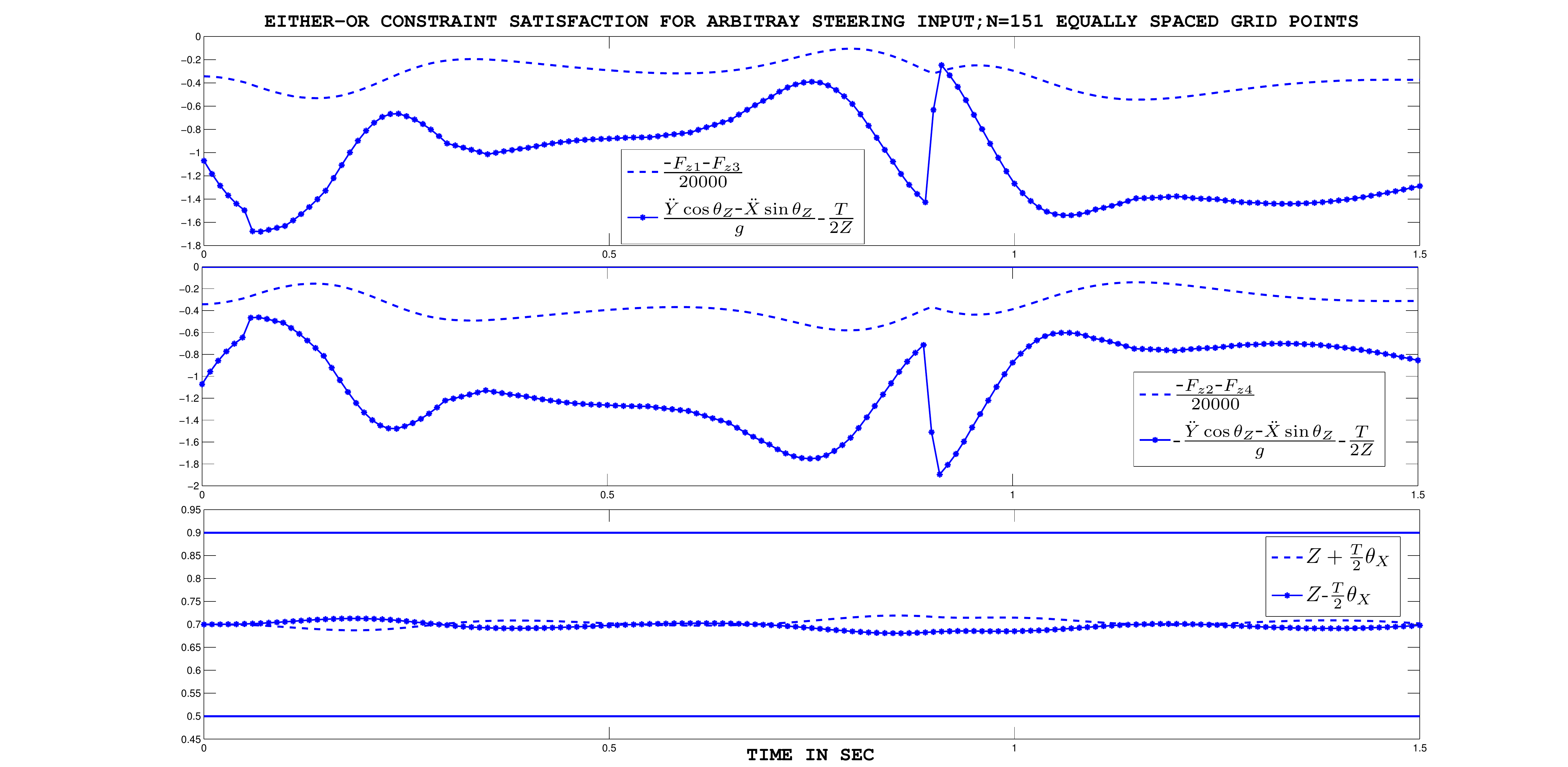}
\caption{Validation: Disjunctive Constraints Satisfaction with an Arbitrary Steering Input
(figure \ref{Fig:ArbSteeringFun}) simulated with input of sensor adapted control forces obtained from
the dynamic optimization with disjunctive constraints and anti-symmetric controls.}
\label{Fig:CON_EO_Arb_Steer}
\end{figure}
\subsection{Robustness Check Against More Arbitrary Steering Inputs}
We use the steering inputs in figures \ref{Fig:NArbS} and \ref{Fig:N2ArbS}. It may be observed that although the steering input rates and values
are higher the controls synthesized with the steering input of figure \ref{Fig:SteeringFun} still works as the control is primary dependent on the yaw
rates which remain comparable in these severe maneuvers. The disjunctive constraints are satisfied in case of both the inputs. We can see
that in figure \ref{Fig:CON_EO_Arb_SteerN1}, the first plot shows that at about 1 s, for a short while only one 
constraint is satisfied, indicating that the control forces are providing the anti-roll moment. 
Although the input in figure \ref{Fig:N2ArbS}
is a more severe maneuver, the plots in figure \ref{Fig:CON_EO_Arb_SteerN2} 
confirm that the disjunctive constraints are satisfied with the 
controls synthesized in section \ref{synth}.
\begin{figure}
\centering
\includegraphics[scale=0.186]{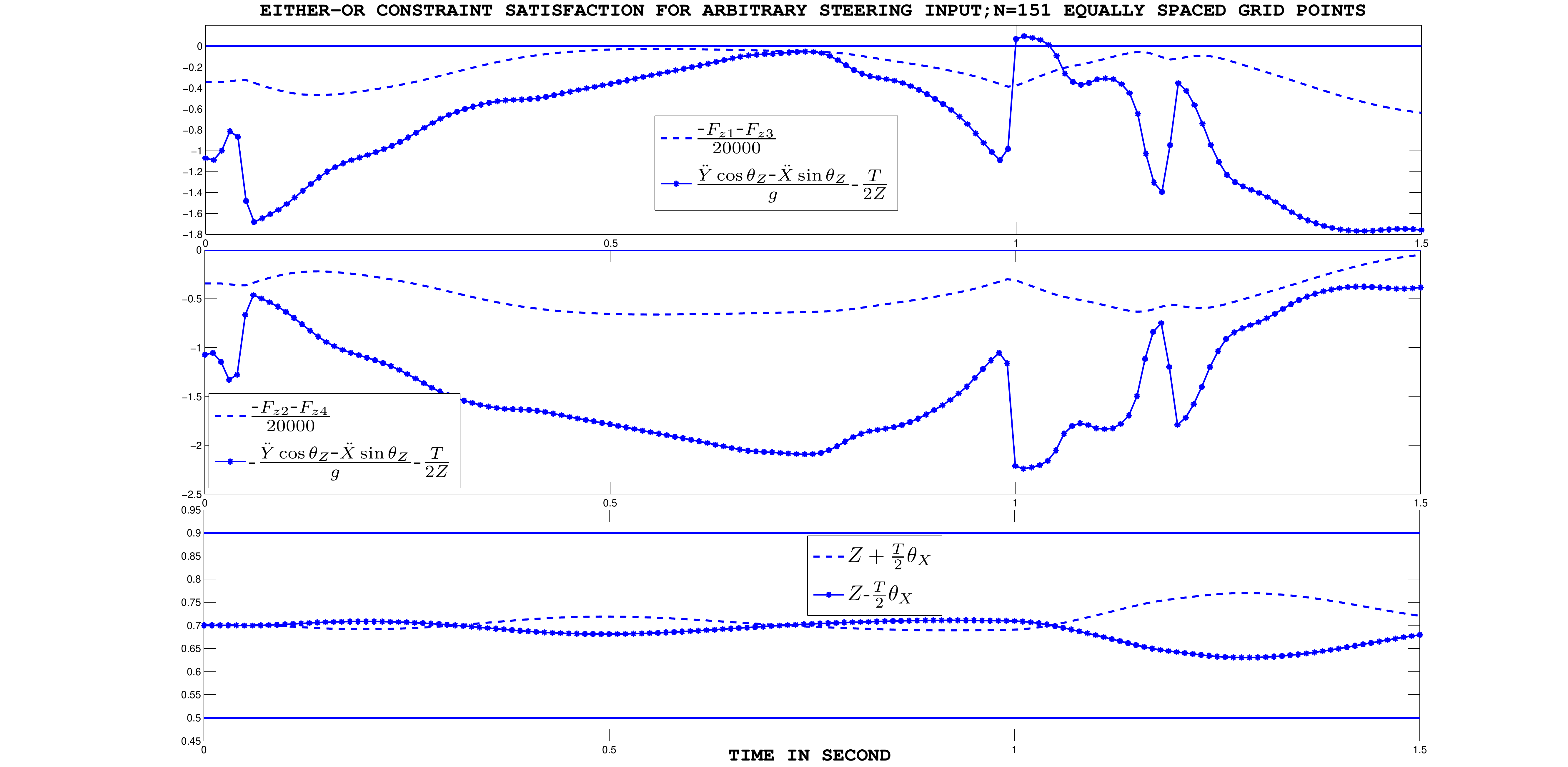}
\caption{Robustness: Disjunctive Constraints Satisfaction with an Arbitrary Steering Input
(figure \ref{Fig:NArbS}) simulated with input of sensor adapted control forces obtained from
the dynamic optimization with disjunctive constraints and anti-symmetric controls.}
\label{Fig:CON_EO_Arb_SteerN1}
\end{figure}
\vspace{-0.2cm}
\begin{figure}
\centering
\includegraphics[scale=0.186]{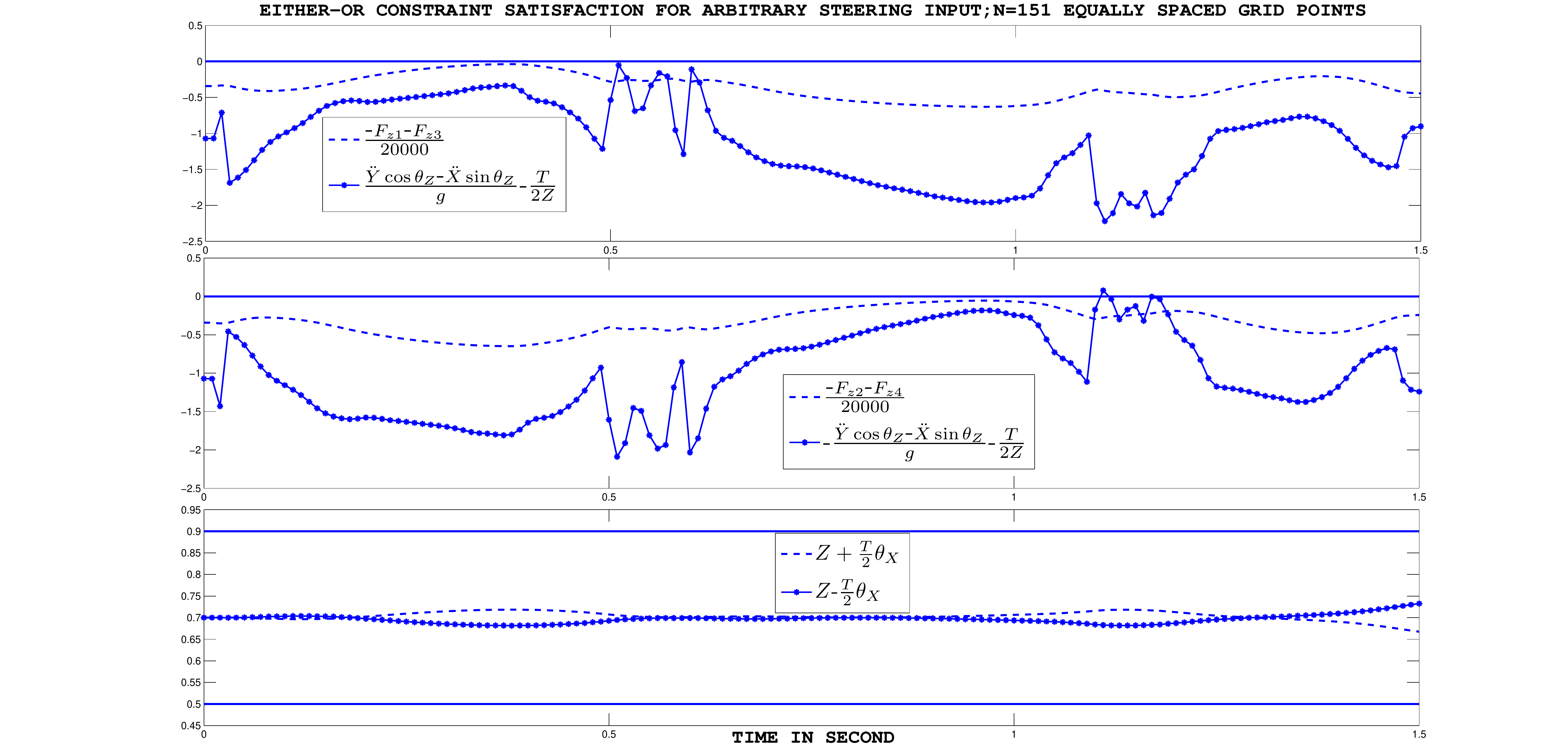}
\caption{Robustness: Disjunctive Constraints Satisfaction with an Arbitrary Steering Input
(figure \ref{Fig:N2ArbS}) simulated with input of sensor adapted control forces obtained from
the dynamic optimization with disjunctive constraints and anti-symmetric controls.}
\label{Fig:CON_EO_Arb_SteerN2}
\end{figure}
\section{Conclusion}
In this work, we have demonstrated that the transcription method 
can be used to solve a non-linear disjunctively constrained problem in vehicle dynamics.
In the process, we have found that the disjunctive constraints enable us to stabilize
vehicles with wheels lifted off, whenever there is a room for doing so, using control forces
comparable to that required for the existing conservative approach of not allowing the wheels to lift off. 
This increases safety under severe maneuver over the conservative approach which is limited to wheels not lifting off the ground. 
Finally we arrived at a simple linear formula proportional to sensor output
data enabling the synthesis of the anti-rollover control forces. 
Future work may be directed toward obtaining smoother independent control forces in the suspensions. 
Other objective functions based on control effort and handling comfort in place of (\ref{Mobj}) need to be explored. 
Additionally, path constraints such as $X(t)=P_1(t)$ and $Y(t)=P_2(t)$ where $P_1(t)$ and 
$P_2(t)$ may be prescribed as known paths in time for increased smoothness and handling comfort.
However, the high index of the resulting differential-algebraic equation model may be a
potential computational difficulty that would need finding effective discretization and optimization strategies.
\section*{Appendix: Parameters Used in Computations}
We list below the simulation data and values of model parameters used in the numerical computations. \\\\
{\bf Time Interval:}~~$t_0=0$ s and $t_f=1.5$ s. \\\\
{\bf  Initial Conditions}\\
$X(0)=0$,  $\dot X(0)=\dot X_0=\frac{200}{9}\textnormal{m/s}$\\
$Y(0)=0$,  $\dot Y(0)=0$\\
$Z(0)=Z_0$,  $\dot Z(0)=0$\\
$\theta_X(0)=0$,  $\dot \theta_X(0)=0$\\
$\theta_Z(0)=0$,  $\dot \theta_Z(0)=0$\\\\
{\bf Parameter Values Used in the Model}\\
$M=1400$ kg\quad(vehicle mass)
$T=1.5$ m \quad(track width)\\
$K=30000$ kg/s$^2$\quad(suspension stiffness)\\
$C=4000$ kg/s\quad(suspension damping)\\
$I_{XX}=1300$ kgm$^2$\quad(roll moment of inertia)\\
$I_{ZZ}=4000$ kgm$^2$\quad(yaw moment of inertia)\\
$h=0.7$ m \quad (height of center of gravity (CG))\\
$a=1.4$ m\quad(longitudinal distance of front axle from CG)\\
$b=1.5$ m\quad(longitudinal distance of rear axle from CG)\\
$g=9.8$ m/s$^2$\quad(acceleration due to gravity)\\
$\mu=1.3$\quad(friction coefficient)\\
$r_{X1}=r_{X2}=a$\quad\\
$r_{X3}=r_{X4}=-b$\quad\\
$r_{Y1}=r_{Y3}=-r_{Y2}=-r_{Y4}=\frac{T}{2}$\quad\\
$\delta_3=\delta_4=0$\quad(steering angle of rear wheels) \\
$\delta_1=\delta_2=\delta(t)$\quad(steering angle of front wheels, see figures \ref{Fig:SteeringFun}-- \ref{Fig:N2ArbS}) \\
$Z_{min}=0.5$ m\quad(minimum height of suspension mount point)\\
$Z_{max}=0.9$ m \quad(maximum height of suspension mount point)\\
$F_{max}=10000$ N \quad(maximum controlling force limit) \\\\
{\bf Constants Used in Tire Force Calculation}\\
$C_T=1.30$,   $\Delta S_h=0$\\
$a_1=-22.1$,  $a_2=1011$,  $a_3=1078$\\
$a_4=1.82$,  $a_5=0.208$,  $a_6=0$\\
$a_7=-0.354$,  $a_8=0.707$
\begin{figure}
\centering
\includegraphics[scale=0.18]{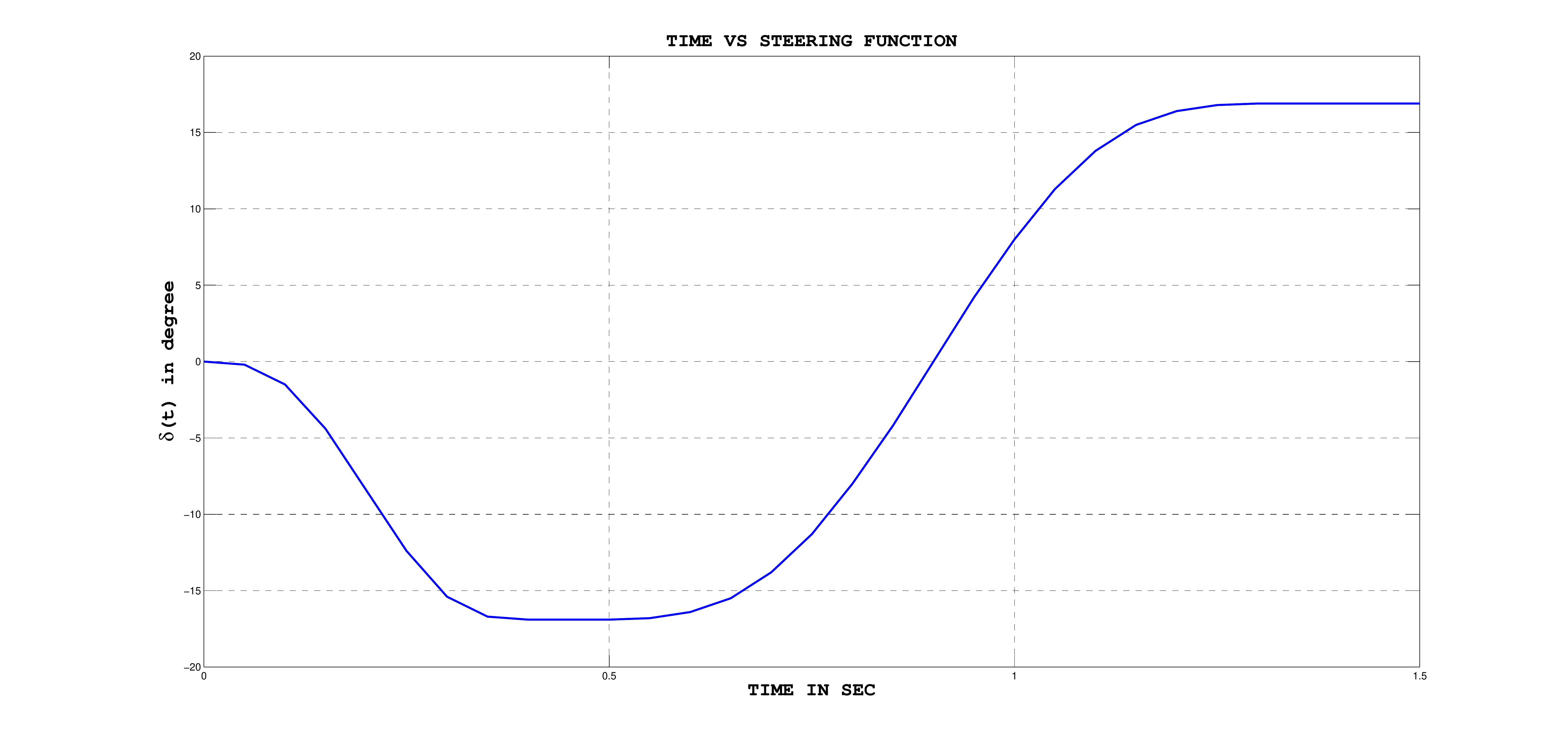}
\caption{Fishhook steering function $\delta_i(t),~i=1,2$ in degrees}\label{Fig:SteeringFun}
\end{figure}
\begin{figure}
\centering
\includegraphics[scale=0.18]{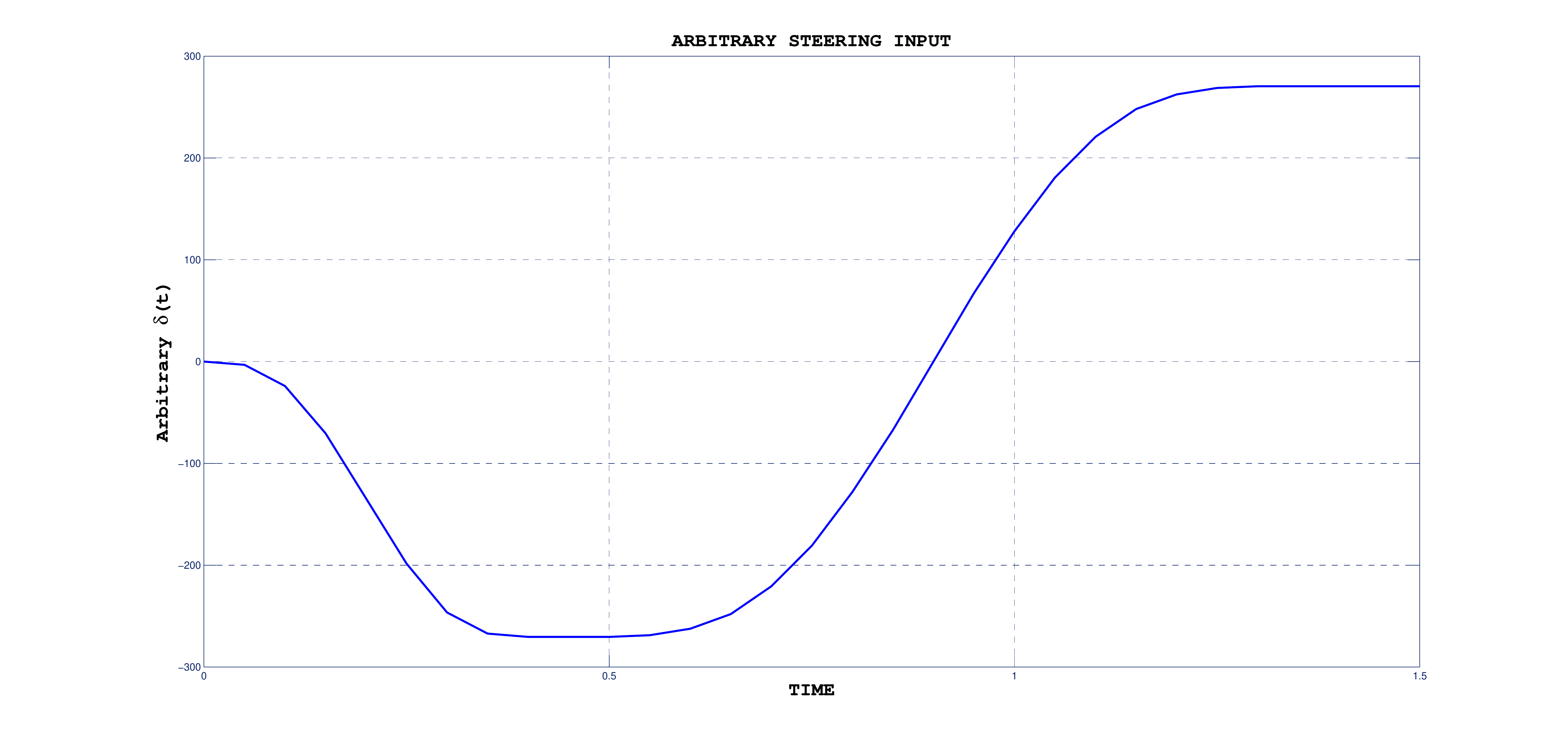}
\caption{Arbitrary steering function $\delta_i(t),~i=1,2$ in degrees. Although similar to the input with which control was synthesized,
the steering rate is faster.}\label{Fig:ArbSteeringFun}
\end{figure}
\begin{figure}
\centering
\includegraphics[scale=0.18]{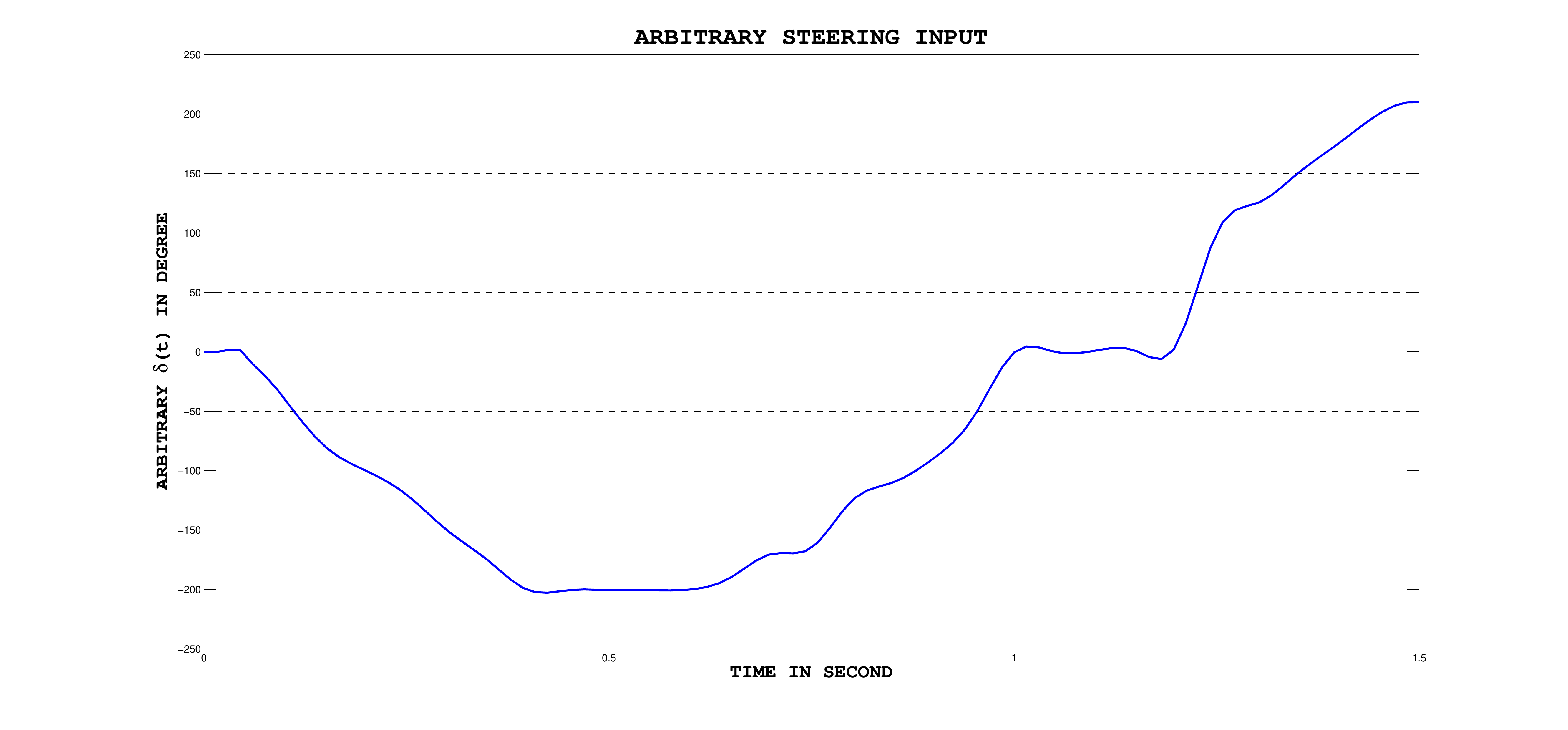}
\caption{Arbitrary steering function $\delta_i(t),~i=1,2$ in degrees. This is a more severe maneuver.}\label{Fig:NArbS}
\end{figure}
\begin{figure}
\centering
\includegraphics[scale=0.18]{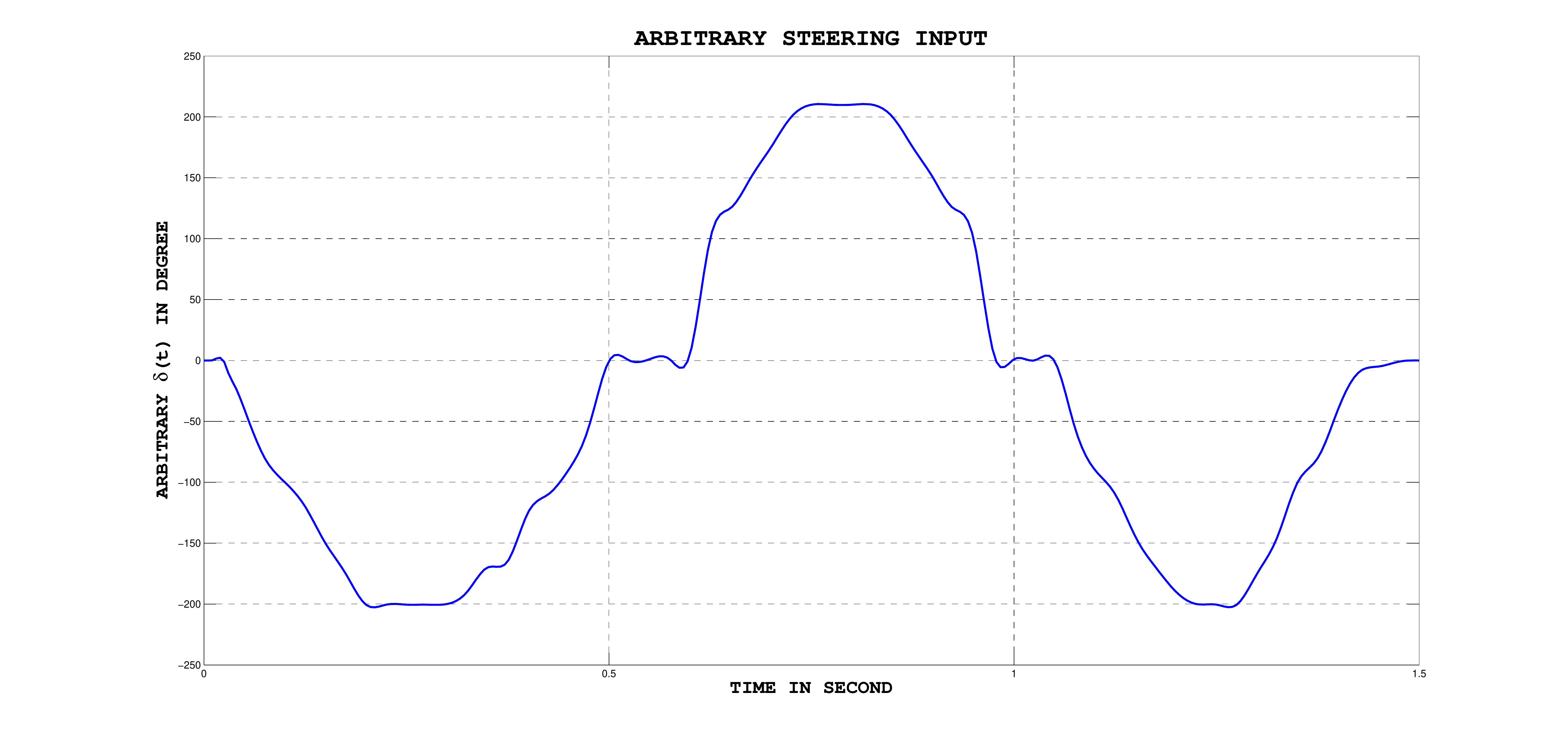}
\caption{Arbitrary steering function $\delta_i(t),~i=1,2$ in degrees. Steering input rates are much faster and the maneuver is even
more severe than that in figure \ref{Fig:NArbS}.}\label{Fig:N2ArbS}
\end{figure}
\newpage
\section*{Acknowledgments} Work of the first author was supported in parts by the CSIR grant $9|79(2368)/2010$ EMR-I
and by the Tata Consultancy Services graduate school fellowship R(II)TCS-Re Schp/2012/2847. The second author's work
was partially supported by the DST grant SR/S4/MS: 683/10 DT. 31.12.2010.

\bibliographystyle{IEEEtran}
\bibliography{references}

\end{document}